\documentclass{article}

\usepackage{amsthm}

\newtheorem{remark}{Remark}[section]
\newtheorem{theorem}{Theorem}[section]
\newtheorem{definition}{Definition}[section]

\newtheorem{corollary}[theorem]{Corollary}
\newtheorem{lemma}[theorem]{Lemma}
\newtheorem{proposition}[theorem]{Proposition}

\usepackage{microtype}
\usepackage{amsmath,latexsym,amssymb}

\usepackage[inline]{enumitem}
\usepackage[usenames,dvipsnames,svgnames]{xcolor}

\usepackage{tikz-cd}
\usetikzlibrary{decorations.pathmorphing}
\usepackage[customcolors]{hf-tikz}
\hfsetfillcolor{gray!20} 
\hfsetbordercolor{white}

\newtheorem{fact}{Fact}

\newcommand{\define}{\triangleq}
\newcommand{\pat}{\rightsquigarrow}

\usepackage{algorithm}
\usepackage[noend]{algorithmic}
\newcommand{\fun}[1]{{\textsl{#1}}}
\algsetup{linenodelimiter=.} 
\algsetup{indent=3em}

 \newcounter{ncomm}%

\newcommand{\W}{\cal W}
\renewcommand{\emptyset}{\varnothing}

\begin{document}

\sloppy

\title{A Polynomial-time Algorithm for Detecting the\\ Possibility of Braess Paradox in Directed Graphs}

\author{Pietro Cenciarelli \and Daniele Gorla \and Ivano Salvo\\
Sapienza University of Rome, Department of Computer Science\\
{\tt cencia@di.uniroma1.it, gorla@di.uniroma1.it, salvo@di.uniroma1.it}}

\maketitle

\begin{abstract}
A  directed multigraph is said {\em vulnerable} if it can generate 
{\em Braess paradox} in Traffic Networks. In this paper,
we give a graph-theoretic characterisation of vulnerable directed multigraphs;
analogous results appeared in the literature only for undirected multigraphs and for a specific family of directed multigraphs. 
The proof of our characterisation also provides an algorithm 
that checks if a multigraph is vulnerable in $O(|V| \cdot |E|^2)$;
this is the first polynomial time algorithm that checks vulnerability for general directed multigraphs.
The resulting algorithm also contributes to another well known problem, i.e. the directed subgraph
homeomorphism problem without node mapping, by providing another pattern graph for which a
polynomial time algorithm exists.
\end{abstract}

\section{Introduction}

{\em Traffic Networks} \cite{braessFormal,BI97} provide a model for studying selfish routing: 
non-cooperative agents travel from a source node $s$ to 
a destination node $t$. Since the cost (or latency) experienced by an agent 
while traveling along a path depends on network congestion
(and hence on routes chosen by other agents), 
traffic in a network stabilizes to the equilibrium of 
a non-cooperative game, where all agents experience the same latency. 
This phenomenon has been defined by Wardrop \cite{War52} in the
context of transport analysis.

For example, consider Fig.~\ref{fig:Wheat}(1), that depicts the so called
 {\em Wheatstone network}.
In the model, edges are labeled by a function (here, the constants 1 and 0, and the identity function $x$)
that specifies the latency of each edge in terms of the flow 
it experiences. So, for example,
the edge of latency 0 is ``ideal'', in the sense that, independently on how much traffic travels along it,
the passage from $b$ to $c$ is instantaneous. By contrast, the delay on passing from $a$ to $b$ and from
$c$ to $d$ is linear in the amount of flow traveling along such edges.
Thus, every small autonomous flow particle $\epsilon$ will try to use the edge $b \rightarrow c$,
because in this way it will experience a delay of $2\epsilon$ (instead of $1+\epsilon$, if it
had avoided that edge).
Consequently, a Wardrop flow of value 1 assigns all the flow to the path $a\,b\,c\,d$ in the picture,
and the overall latency is 2.   

In traffic networks, a well known and counterintuitive phenomenon is {\em Braess paradox} \cite{braessFormal,braessOriginal}, 
that originates when latency at Wardrop equilibrium decreases because of removing edges.
The Wheatstone network is a minimal example of Braess paradox:
Fig.~\ref{fig:Wheat}(2) shows its optimal subnet. There, a Wardrop flow of value 1 assigns 1/2
to both paths in the network ($a\,b\,d$ and $a\,c\,d$), thus obtaining a latency of 
3/2.
%

\begin{figure}[t]
\begin{tabular}{ccc}
\begin{minipage}{0.29\textwidth}
\center
\begin{tikzcd}[column sep=.8cm,row sep=.4cm]
   & b \ar[rd, "1"]\ar[dd,"0"]\\
a \ar[ur,"x"] \ar[dr,"1" '] && d\\
   & c \ar[ru,"x" ']\\
   &(1)
   \\
\end{tikzcd}
\end{minipage}
&
\begin{minipage}{0.29\textwidth}
\center
\begin{tikzcd}[column sep=.8cm,row sep=.4cm]
   & b \ar[rd,"1"]\\
a \ar[ur,"x"] \ar[dr,"1" '] && d\\
   & c \ar[ru,"x" ']\\
   &(2)
   \\
\end{tikzcd}
\end{minipage}
&
\begin{minipage}{0.33\textwidth}
\center
\begin{tikzcd}[column sep=.8cm,row sep=.4cm]
& b \ar[dd]\ar[dr] & \\
a \ar[ru]\ar[rd] & & d\\ 
& c \ar[ru]\\
&(3)
\\
\end{tikzcd}
\end{minipage}
\vspace{-3mm}
\end{tabular}
\caption{(1) The Wheatstone network; (2) Its optimal subgraph; (3) The graph $\W$.}
\label{fig:Wheat}
\vspace{-3mm}
\end{figure}
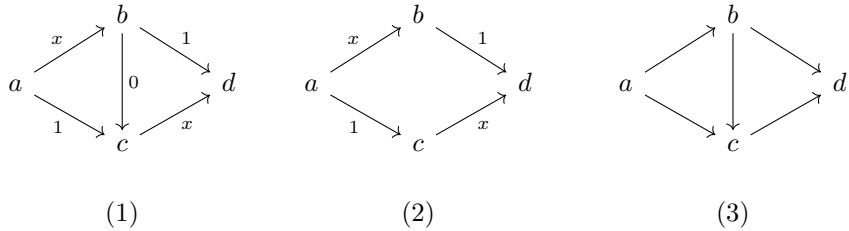

Braess paradox has been studied for decades. The results that are most strongly related to ours
start with \cite{Rough06}, where it is shown that, given a multigraph, a latency
function on its edges and the total amount of flow, it is NP-hard to prove
whether the resulting net suffers from the Braess paradox or not. 
An intriguing question raised in \cite{Rough06} (Open Question 1 in Sect.\,6.1) is to study
Braess paradox from a graph-theoretical perspective, i.e. by considering only a graph and studying
whether it admits instances that generate the paradox. This property of a graph has been called {\em vulnerability}.

A characterisation of vulnerable undirected multigraphs is presented in~\cite{Milch06}, where
it is proved that an undirected graph is vulnerable if and only if it is not series-parallel \cite{RS42}.
This characterisation only holds for graphs where every node and every edge lie on at least one 
simple $st$-path. 
Later, in \cite{ChenEtal15}, the same characterisation is proved for directed multigraphs
that satisfy the very same condition, called {\em irredundancy} therein. 
However, while checking redundancy and calculating the maximal irredundant 
subgraph can be efficiently done in undirected
graphs \cite{CinesiFull}, the same does not hold for directed graphs. Indeed,
in  \cite{CinesiFull} the authors hint at the difficulty of this problem and note that their solution is
tractable only for planar digraphs. They also conjecture that the problem of recognizing
vulnerable digraphs is untractable, in general.

\paragraph{Contribution}
In this paper, we disprove this conjecture. First of all, we prove NP-hardness 
of checking and removing redundancy in directed multigraphs, thus formalizing
the informal claim in \cite{CinesiFull}. Then, we provide a graph-theoretic characterisation
of vulnerable directed nets, by
proving that a directed multigraph is vulnerable if and only if it
contains 
the graph $\W$ underlying the Wheatstone network and depicted in Fig.~\ref{fig:Wheat}(3).   
Our constructive proof 
provides an algorithm to check vulnerability, with an execution time that is $O(|V| \cdot |E|^2)$. 
This is the first polynomial time algorithm we are aware of for checking vulnerability 
of general directed multigraphs and fixes the issue left open in \cite{CinesiFull}.

Finally, our characterisation and the resulting polynomial time
algorithm contribute to another well known problem: the directed subgraph
homeomorphism problem without node mapping \cite{For1980}.
This problem consists in fixing a pattern graph ($\W$ in our case) and try to find
an homeomorphic copy of it within an input graph. In \cite{For1980}, it is 
proved that the general problem is NP-hard unless the pattern
graph has all nodes with indegree at most 1 and outdegree at most 2, or indegree at most 2 and 
outdegree at most 1 (as for $\W$); pattern graphs of this kind, for which 
a polynomial time algorithm exists, are known in the literature \cite{HU72}. 
However, as far as we know, no polynomial time algorithm for all this class of pattern
graphs has been devised so far.
Thus, our work also contributes to this research line, by providing another pattern 
graph for which a polynomial time algorithm 
exists.

\vspace{-2mm}
\paragraph{Related work}
As we already said, the most strongly related papers are \cite{ChenEtal15,CinesiFull,Milch06}.

In \cite{Milch06}, vulnerable undirected multigraphs are characterised as the non-series-parallel
ones that, in turn, are those containing $\W$ as homeomorphic subgraph. 
Moreover, this work also characterises the undirected networks where all Wardrop equilibria are
weakly Pareto efficient; these turn out to be the nets with linearly independent routes (i.e., those 
nets where every $st$-path has at least one edge that belongs only to that path). Linear independence
of routes is also shown to be the characterisation of those undirected nets that are efficient 
under heterogeneous players (i.e., nets where different players can have different
latencies on the same edges). Furthermore, linear independence is characterised by
not having as homeomorphic subgraph any of three elementary nets (one of which is $\W$).

In \cite{ChenEtal15,CinesiFull}, directed networks are considered, both in their single commodity version
(i.e., with just one pair of source and target) and in the multicommodity one (i.e., with many such pairs).
Assuming irredundancy of the net (i.e., that every node and edge lies on at least one $st$-path, for some
$st$ pair), they prove that, for single commodity nets, vulnerability coincides with not being series-parallel
(like in \cite{Milch06}); the result is then properly generalised to multicommodity nets. 
As we show in this paper, checking
irredundancy of a net and calculating the maximal irredundant subnet are computationally difficult
problems. Indeed, in \cite{CinesiFull} an algorithm for checking vulnerability is provided only for undirected
and for planar directed nets.

An orthogonal bunch of works \cite{LRTW11,Rough06,RT02} has been devoted to the complexity 
of estimating the Braess ratio (i.e., the maximum ratio between the Wardrop latency of $G$ and of any its 
subnet) and the price of anarchy (i.e., the worst-case ratio between
the values of any Wardrop flow and of the optimal one). These works show that both measures 
have very strong inaproximability, both for the single and for the multicommodity 
directed scenarios.

Other works have been carried out to graph-theoretically characterise networks with similar kinds of
games \cite{EFM09,HN03,HM15,Milch05,Milch15}. These works differ from ours in the kind of efficiency one
aims at or in the model of the game. For example, in \cite{Milch05} the aim is to have all players
with the same latency in all equilibria, whereas in \cite{EFM09} the aim is to characterise nets
whose equilibria all minimize the maximum latency of every path, in a fremework where 
every player can only choose one path and send along it just one information unit.

Finally, an orthogonal paper is \cite{matroid15}, where Braess paradox
is generalized to all congestion games; structures that do not suffer of Braess paradox 
are then characterised in terms of matroids. Their elegant result, differently from ours,
it is not directly related to graph-theoretic concepts, nor it provides any algorithmic procedure.

\vspace{-2mm}
\paragraph{Organization of the paper}
We start in Section 2 by giving the basic notions on traffic networks and Braess paradox; this
leads to the definition of vulnerability.
Then, in Section 3 we show that the only existing characterisation of vulnerability for directed
multigraphs \cite{ChenEtal15,CinesiFull} cannot yield a polynomial algorithm for general multi-digraphs, 
unless P = NP.
In Section 4 we provide our characterisation. In Section 5 we show how this characterisation
yields a polynomial time algorithm and then show how this can be used for the directed subgraph
homeomorphism problem. In Section 6 we conclude the paper.

\section{Traffic Networks, Braess Paradox and Vulnerability}
\label{sec:vuln}

We start by providing the necessary background, by
essentially following the presentation in~\cite{Rough06}.
%
A \emph{directed multigraph} (or  \emph{multi-digraph}) $G = (V,E)$ consists of a
set $V$ of \emph{vertices} (or \emph{nodes}) and a set $E$ of \emph{edges}. 
Every edge relates a pair of vertices; if $e$ relates $(u,v)$, we say that $e$ is
an \emph{output} of $u$ and an \emph{input} of $v$.
We denote with $in(x)$ and $out(x)$ the edges entering into $x$ and exiting from $x$, respectively.
When the name of an edge is not relevant but only its extremes are,
we denote an edge that relates $(u,v)$ as  $u \rightarrow v$.
 
A \emph{path} is a sequence $u_1 e_1 u_2 \ldots u_{n-1} e_{n-1} u_n$ (for $n \geq 1$) of nodes and edges 
such that $e_i$ relates $(u_i,u_{i+1})$, for all $i<n$;
$u_1$ and $u_n$ are called {\em extremes} and $u_2,\ldots,u_{n-1}$ are called {\em internal nodes}.
When only the extremes of the path $p$ are relevant, we shall write
$u_1 \stackrel p \pat u_n$ (or simply $u_1 \pat u_n$ if we ignore the name given to the path).
If each $u_i$ appears just once in $p$, we say that $p$ is {\em simple} (or {\em acyclic}).
We fix a \emph{source} node, written $s$, and a \emph{target} node, written $t$.
The resulting triple $(G,s,t)$ will be called {\em net}; we shall sometimes
simply write it $G$, i.e., without specifying the source and target (that we leave understood). 
An \emph{st-path} is a 
path from $s$ to $t$. The set of $st$-paths in $G$ is denoted by $P(G)$; 
the set of 
simple $st$-paths is denoted by $SP(G)$. 
We say that a net is $st$-{\em connected} if all vertices of $G$ belong to 
an $st$-path.
%
%

A \emph{flow} for a net $(G,s,t)$ is a function 
$\varphi:SP(G) \rightarrow \mathbb{R}^+$.
The value of a flow is the sum of the values sent over all paths. 
A flow induces a unique flow $\varphi(e)$ on edges: for any edge $e\in E$, 
$\varphi(e)=\sum_{p\in SP(G): e\in p}\varphi(p)$. 
%
A {\em latency function} $l_e:\mathbb{R}^+\rightarrow \mathbb{R}^+$ 
assigns to each edge $e$ a latency that depends on the flow on it;
as usual, we only consider continuous and non-decreasing latency functions.
The latency of an $st$-path $p$ under a flow $\varphi$ is the sum of the latencies of 
all edges in the path under $\varphi$, i.e., $l_p(\varphi)=\sum_{e\in p} l_e(\varphi(e))$.
If $H$ is a subgraph of $G$, we denote with 
$l|_H$ the restriction of the latency function $l$ on the edges of $H$.

Given a net $G$, a real number $r\in\mathbb{R}^+$ and a latency function $l$, 
we call the triple $(G, r, l)$ an {\em instance}. A flow $\varphi$ is {\em feasible} for 
$(G, r, l)$ if the value of $\varphi$ is $r$. Notice that, since we do not have any
constraint on edges or vertices, every $r$ admits at least one feasible flow.

A feasible flow $\varphi$ for $(G, r, l)$ is at {\em Wardrop equilibrium} 
(or is a {\em Wardrop flow})\footnote{
	Here we use as definition the characterisation of Wardrop flow
	put forward by Proposition 2.2 in \cite{Rough06}.
} if, for all pairs $p, q$ of $st$-paths  
such that $\varphi(p)>0$, we have $l_p(\varphi)\leq l_q(\varphi)$.
In particular, this implies that, if $\varphi$ is a Wardrop flow, 
all $st$-paths to which $\varphi$ assigns a positive flow
have the same latency. It is known \cite{Rough06} that every
instance admits a Wardrop flow and that different Wardrop flows for the same
instance have the same latency along all $st$-paths with a positive flow. Thus,
we denote with $L(G, r, l)$ the latency of all $st$-paths 
with positive flow at Wardrop equilibrium. In the special case where $r = 0$,
we let $L(G, r, l)$ be 0.

{\em Braess paradox} \cite{braessFormal,braessOriginal} originates when latency 
at Wardrop equilibrium decreases because of removing edges (or equivalently, by raising 
the latency function on edges): an instance $(G, r, l)$ suffers from Braess paradox 
if there is a subgraph of $G$ with a lower latency.

In \cite{Rough06} it is shown that, given an instance $(G,r,l)$, it is NP-hard to prove
whether it suffers from the Braess paradox or not. An intriguing question raised in \cite{Rough06} is to study
the problem from a graph-theoretical perspective, i.e. by only having $G$ and studying
whether $G$ admits instances that generate the paradox.
This leads to the following definition \cite{Rough06}.
\begin{definition}
	A net $G$ is {\em vulnerable} if there exist a value $r$, a latency function
	$l$ and a subgraph $H$ of $G$ such that $L(G, r, l)>L(H, r, l|_H)$.
\end{definition}

\section{On the Complexity of Irredundancy in Directed Nets}
\label{sec:IRRED}

All characterisations of vulnerable nets we are aware of are for {\em irredundant} nets,
that are nets containing only irredundant vertices and edges.
As defined in \cite{ChenEtal15} (and derived from \cite{Milch06}), a vertex or an edge is {\em irredundant}
if it appears in a simple $st$-path, and {\em redundant} otherwise. 
Redundant vertices and edges can be safely ignored when studying
the Braess phenomenon in a multigraph: 
they will never be touched by a Wardrop flow, since only acyclic paths have a positive flow. 
The characterisation given by \cite{ChenEtal15} provides a polynomial time
algorithm for checking vulnerability for irredundant nets, 
where it suffices to check whether the net is series-parallel
(this can be done in linear time, by following \cite{VTL82}).
If a polynomial reduction from a redundant net to an irredundant subnet with the same simple paths existed, then a polynomial time algorithm would exist also for redundant nets.
As we now show, this is not possible unless P = NP. 
Furthermore, we will show that also verifying if a net is irredundant 
is an NP-complete problem.
Therefore, checking if a net satisfies the preconditions prescribed 
by the algorithm in \cite{CinesiFull} cannot 
be done in polynomial time, unless P = NP.
%

First, we observe that irredundancy can be checked by only focusing on edges.

\begin{fact}
If every edge of $(G,s,t)$ is irredundant, then every node is irredundant.
\end{fact}
\begin{proof}
Assume that $v$ is redundant: this means that every $st$-path containing $v$ is not simple.
Fix one of these paths $p$ and consider the edges of $p$ that are input and output of $v$:
they cannot be irredundant, otherwise also $v$ would be. Contradiction.
\end{proof}

We start with an easy characterisation of irredundant edges. By exploiting it, we then move to the problems of:
(1) checking whether an edge is irredundant; (2) deriving from any net its maximal irredundant subnet 
with the same 
acyclic paths; and (3) checking whether a net is irredundant. We will prove that all these three problems are NP-hard.

\begin{fact}
\label{lem:char-irr-edge}
An edge $e = u\rightarrow v$ is irredundant in the net $(G,s,t)$ if and only if 
there exist two node-disjoint paths $s \leadsto u$ and $v\leadsto t$.
\end{fact}
\begin{proof}
The `if' part is trivial: the $st$-path $s \leadsto u\, e\, v \leadsto t$ is acyclic and contains $e$.
For the `only if' part, by definition, there exists a simple $st$-path $v_0\, e_0\, v_1\, e_1 \ldots v_n$ (for $v_0 = s$, $v_n = t$ and $n \geq 1$) such that $e = e_{i}$, for some $i \in \{0,\ldots,n-1\}$.
Thus, $s \leadsto u$ is $v_0 \,e_0\, v_1 \ldots v_i$ and
$v \leadsto t$ is $v_{i+1}  \ldots v_n$: these paths 
are node-disjoint because the original
path was simple.
\end{proof}

\newcommand{\ewirr}{{\bf EW-$st$-IRR}}
\newcommand{\twodpp}{{\bf 2-DPP}}
\newcommand{\dered}{{\bf MIS}}

\begin{definition}
\label{def:irr-edge}
{\em Edge-wise $st$-irredundancy} (denoted as \ewirr) is the problem of deciding whether, given a multi-digraph
$G = (V,E)$, two distinct nodes $s,t \in V$ and an edge $e \in E$, the edge $e$ is irredundant in the net $(G,s,t)$.
\end{definition}

We first observe that \ewirr\ belongs to NP. By definition,
a certificate for the edge $e$ in $G$ to be irredundant 
is a simple path from $s$ to $t$ containing $e$.

\begin{proposition}
\label{prop:ewirr-NPhard}
\ewirr\ is NP-hard.
\end{proposition}
\begin{proof}
Thanks to Fact \ref{lem:char-irr-edge},
we show that \ewirr\ can be easily used to solve \twodpp\ (the {\em 2-disjoint paths problem}) \cite{MP93}:
\begin{quote}
\twodpp: given a directed graph $G = (V,E)$ and four distinct vertices $x,y,w,z$, 
decide whether there exist $x \leadsto y$ and $w \leadsto z$ node-disjoint.
\end{quote}
The polynomial reduction is the following.
If $G$ contains an edge $y \rightarrow w$, it suffices to check whether such an edge is $xz$-irredundant.
Otherwise, we consider the graph $G' = (V, E \uplus \{e\})$, where $e$ is a new edge relating $(y,w)$,
and check whether $e$ is $xz$-irredundant in $G'$.
\end{proof}

NP-completeness of checking whether an edge is irredundant in a net will now be used
to show NP-hardness of the problem of extracting from a given net its maximal equivalent 
(with respect to vulnerability) irredundant subnet. Indeed, we shall prove that a net and 
any of its subnets are equivalent with respect to vulnerability if they have the same set of simple paths
(see Proposition \ref{lemma:equivalentGraphs} later on).

\begin{definition}
\label{def:dered}
The {\em maximal irredundant subnet} problem (denoted as \dered) is the problem of calculating, 
given a net $(G,s,t)$, its maximal irredundant subnet $(G',s,t)$ such that $SP(G)=SP(G')$. We denote
$G'$ by $\fun{MIS}(G,s,t)$
\end{definition}

\begin{lemma}
\label{lem:dered-unique}
$\fun{MIS}(G,s,t)$ is unique.
\end{lemma}
\begin{proof}
Redundant vertices and edges are univocally determined, once fixed $G$, $s$ and $t$. 
Moreover, removing redundant edges does not change the set of simple paths in $G$. 
So, it suffices to delete them from $G$ in order to obtain a subnet $(G',s,t)$ that is irredundant
but has the same simple $st$-paths as $G$. Indeed, in passing from $G$ to $G'$, we have 
only removed edges and vertices that either do not belong to any $st$-path (viz., those vertices that
are unreachable from $s$ or that cannot reach $t$, and edges incident to them) or belong only to
cyclic $st$-paths. To show uniqueness, notice that: if we remove other edges/vertices, the resulting net will
loose some acyclic $st$-path; if we remove not all these edges/vertices, the resulting net will still be
redundant.
\end{proof}

\begin{proposition}
\label{lem:dered-np}
\dered\ is NP-hard.
\end{proposition}
\begin{proof}
Thanks to Lemma \ref{lem:dered-unique},
\dered\ can be easily used to solve \ewirr:
an edge is $st$-irredundant for $(G,s,t)$ if and only if it appears in $\fun{MIS}(G,s,t)$.
\end{proof}

Finally, also the 
problem of checking whether a given net is 
irredundant is an NP-complete problem. In principle, this problem could be computationally simpler:
there could exist some characterisation that provides a polynomial test to check if a net is redundant or not
but without giving any hint on which edges are redundant and which ones are not.

\newcommand{\gred}{{\bf $st$-IRR}}

\begin{definition}
The {\em $st$-irredundancy} problem (denoted as \gred) is the problem of deciding whether, given a multi-digraph
$G = (V,E)$ and two distinct nodes $s,t \in V$,  the net $(G,s,t)$ is irredundant.
\end{definition}

We first observe that \gred\ belongs to NP. A certificate for $G$ to be irredundant are $|E|$ simple $st$-paths 
such that the $i$-th path contains the $i$-th edge of $E$. 
We now reduce \ewirr\ to \gred: given a net $G$ and an edge $u\rightarrow v$, 
we build a new net $(G^*,s^*,t^*)$ that is redundant if and only if $u \rightarrow v$ is redundant in $(G,s,t)$. 
\medskip

\begin{figure}[t]
\begin{center}
\begin{tikzcd}[column sep=.4cm,row sep=.3cm]
&&&&&& \tikzmarkin{a1} &  \textcolor{gray!110}{G'} & u' 
\ar[rrrrddddddd, to path={..controls +(7,-.1) and +(1,0)..(\tikztotarget)}]
\\
&&&&&& s' &&&& t' \ar[rrdd, bend left=30]
\\
&&&&&&& \ar[dlll,Rightarrow] & v' &\ & \tikzmarkend{a1}  &
\\
&&&& r' \ar[uurr,bend left=20]  \ar[rrrrddddddd, bend left=20]  &&&&&&&& a'' \ar[dddlllllllll,bend right=10] \ar[ulll,Rightarrow]
\\
&&& z' \ar[ur] \ar[llld, bend right=20] &&&&&&&&&\ 
\\
t^* && s^*  \ar[ru]\ar[dr]
\\
&&& z''  \ar[dr] \ar[lllu, bend left=20]
\\
&&&& r'' \ar[ddrr, bend right=20] \ar[rrrruuuuu,bend right=10]  &&&&&&&& a' \ar[uuulllllllll,bend left=10]
\ar[dlll,Rightarrow]
\\
&&&&&& \tikzmarkin{a2} & \ar[ulll,Rightarrow] & u'' 
\ar[rrrruuuuu,to path={..controls +(5,-5) and +(2,0)..(\tikztotarget)}] &\ 
\\
&&&&&& \ s'' &&&& t''  \ar[rruu,bend right=10]
\\
&&&&&&& \textcolor{gray!110}{G''} & v'' &&\tikzmarkend{a2} 
\\
\end{tikzcd}
\vspace*{-2.75cm}
\caption{The construction for reducing \ewirr\ to \gred\ (a thick arrow from a node to a grey part representing a
graph means $|V|$ edges from that node to {\em every} node of the graph; and vice versa).}
\label{fig:NPreducrtion}
\end{center}
\vspace{-4mm}
\end{figure}

Intuitively, our reduction is based on the following ideas: 
(1) if $u\rightarrow v$ is irredundant, then there exists a simple path $p$ in $G$ of the form 
$s\leadsto u\rightarrow v\leadsto t$, and  
(2) for each vertex $x$, adding new edges of the form $u\rightarrow x$ and $x\rightarrow v$ 
does not make irredundant the edge $u\rightarrow v$ in $G^*$ if it is redundant in $G$. 
By using such additional edges, if $u\rightarrow v$ is irredundant in $G$, 
we can make irredudant any edge $x\rightarrow y$ of $G$  
thanks to the path $s\leadsto u\rightarrow x \rightarrow y \rightarrow v\leadsto t$ that is simple provided that 
$x$ and $y$ do not occur in $p$. 

To correctly handle the case in which $x$ or $y$ occur in $p$, we build a net $(G^*, s^*, t^*)$ 
that consists of  
two copies $G'$ and $G''$ of $G$ obtained by decorating each node in $G$ with $'$ and $''$, respectively. 
Each edge $x'\rightarrow y'$ in $G'$ (resp. $x''\rightarrow y''$ in $G''$) 
will be irredundant in $G^*$ whenever $u\rightarrow v$ is irredundant in $G$
thanks to a path of the form 
$s^*\leadsto s''\leadsto u''\leadsto x'\rightarrow y'\leadsto v''\leadsto t'' \leadsto t^*$
(resp. $s^* \leadsto s'\leadsto u'\leadsto x'' \rightarrow y'' \leadsto v'\leadsto t'\leadsto t^*$)
obtained by adding in $G^*$ paths $u'\leadsto x''$, $u''\leadsto x'$, $x'\leadsto v''$, and $x''\leadsto v'$, 
for all nodes $x'\in V'$, $x''\in V''$.


To ensure that redundancy of $u\rightarrow v$ in $G$ implies redundancy of $u'\rightarrow v'$ and $u''\rightarrow v''$ 
in $G^*$, we will define $G^*$ in such a way that a simple path entering in $G'$ via $s'$ (resp. in $G''$ via $s''$) 
can reach the target $t^*$ of $G^*$ only  
through the node $t'$ (resp. $t''$), possibly after {\em only one} detour in $G''$ (resp. $G'$).
This is guaranteed by additional nodes $z',z'',r',r'',a',a''$ and their incident edges (see Fig. \ref{fig:NPreducrtion}). 
A simple path can enter in $G'$ (resp. $G''$) only through nodes $z'$ and $r'$ (resp. $z''$ and $r''$): this implies that this path can reach $t^*$ 
only through the path $t'\rightarrow a''\rightarrow z''\rightarrow t^*$ (resp. $t''\rightarrow a'\rightarrow z'\rightarrow t^*$). As a consequence, this path can enter 
$G''$ (resp. $G'$) only through $u'\rightarrow a'\rightarrow x''$ (resp. $u''\rightarrow a''\rightarrow x'$) for some $x''\in V''$ (resp. $x'\in V'$) and it must eventually come back into $G'$ (resp. $G''$) through the path 
$y''\rightarrow r''\rightarrow v'$ (resp. $y'\rightarrow r'\rightarrow v''$) for some $y''\in V''$ (resp. $y'\in V'$).  
 

\begin{theorem}
\label{lem:dered-np}
\gred\ is NP-hard.
\end{theorem}
\begin{proof}
Given an instance for \ewirr\ (i.e., a multi-digraph $G = (V,E)$, two distinct vertices $s,t \in V$ and $u \rightarrow v \in E$),
we first observe that edges entering into $s$ or exiting from $t$ are trivially redundant; so, we can assume that
$v \neq s$ and $u \neq t$. Now,
we build a new net $G^*$ such that $(G^*,s^*,t^*)$ is irredundant if and only if $u \rightarrow v$ is irredundant in $(G,s,t)$.
The construction is as follows:
\begin{itemize}
\item Create two isomorphic node-disjoint copies of $G$, call them $G' = (V',E')$ and $G'' = (V'',E'')$, 
with $V' = \{x' : x \in V\}$, $V'' = \{x'' : x \in V\}$, 
$E' = \{e' : e \in E\}$ and 
$E'' = \{e'' : e \in E\}$, where $e'$ and $e''$ relate $(x',y')$ and $(x'',y'')$ respectively, if $e$ relates $(x,y)$.
\item Define $G^*$ as $(V^*,E^*)$, where  
$$
\begin{array}{ll}
V^*\ = & V' \uplus V'' \uplus \{s^*, z', z'',a', a'', r', r'',t^*\}
\vspace*{.3cm}
\\
E^*\ = 
& (E'\ \setminus (in(s') \cup out(t'))) 
\ \cup\ (E''\ \setminus (in(s'') \cup out(t'')))
\vspace*{.2cm}\\
& \cup\ \{s^* \rightarrow z', z' \rightarrow r', r' \rightarrow s',
t' \rightarrow a'', a'' \rightarrow z'', z'' \rightarrow t^*,
\\
& \hspace*{.55cm} s^* \rightarrow z'', z'' \rightarrow r'', r'' \rightarrow s'',
t'' \rightarrow a', a' \rightarrow z', z' \rightarrow t^* \}
\vspace*{.2cm}\\
& \cup\ \{u' \rightarrow a', r'' \rightarrow v'\}
\ \cup\ \bigcup_{x'' \in V''} \{a' \rightarrow x'',x'' \rightarrow r''\}
\vspace*{.2cm}\\
& \cup\ \{u'' \rightarrow a'', r' \rightarrow v''\}
\ \cup\ \bigcup_{x' \in V'} \{a'' \rightarrow x',x' \rightarrow r'\}
\end{array}
$$
\end{itemize}
The graphical representation of this construction is given in Fig. \ref{fig:NPreducrtion}.

\paragraph{Irredundancy:}

Let us first assume that $u \rightarrow v$ is irredundant in $(G,s,t)$ and show that every edge in $G^*$ is
irredundant. By definition, irredundancy of $u \rightarrow v$ implies the
existence of a simple $st$-path in $G$ that contains it; let $s \stackrel{p_1}\leadsto u \rightarrow v 
\stackrel{p_2}\leadsto t$ be such a path. 
First, note that no edge belonging to $in(s) \cup out(t)$ can appear in this path, since all these
edges are redundant; for this reason, we have excluded them in $G^*$.

To show that a generic edge $x' \rightarrow y' \in E'\ \setminus (in(s') \cup out(t'))$ is irredundant,
let us consider the path
$s^* \rightarrow z'' \rightarrow r'' \rightarrow s'' \stackrel{p_1''}\leadsto u'' \rightarrow a'' \rightarrow 
x' \rightarrow y' \rightarrow r' \rightarrow v'' \stackrel{p_2''}\leadsto t'' \rightarrow a' \rightarrow z' \rightarrow t^*$,
where $p_1''$ and $p_2''$ denote $p_1$ and $p_2$ within $G''$.
By considering such paths for every $x' \rightarrow y' \in E'$, all edges in
$\{s^* \rightarrow z'', z'' \rightarrow r'', r'' \rightarrow s'' , u'' \rightarrow a'' , a'' \rightarrow x', 
y' \rightarrow r', r' \rightarrow v'' , t'' \rightarrow a', a' \rightarrow z', z' \rightarrow t^*\}$ are proved
irredundant, for every $x', y' \in V'$.

The remaining edges of $E^*$ are ``dual'' (i.e. have $'$ and $''$ swapped) of the ones considered so far
and can be proved irredundant in $G^*$ by considering the dual version 
(i.e. with $'$ and $''$ swapped everywhere)
of the path above.

\paragraph{Redundancy:}

Let us now assume that $u \rightarrow v$ is redundant in $(G,s,t)$ but, by contradiction, 
assume that there is a simple $s^*t^*$-path $p$ in $G^*$ that contains $u' \rightarrow v'$. 
Let us now yield a contradiction by reasoning 
on how $p$ has entered into $G'$ for the first time:
\begin{enumerate}

\item {\em Through $r'$:} We first observe that $r'$ can only be reached
through $z'$: the only edges different from $z' \rightarrow r'$ entering into $r'$ are those of
the form $x' \rightarrow r'$, for $x' \in V'$; but using these edges would mean that $p$ has already
entered into $G'$. Second, we observe that the only way for reaching $t^*$ with a simple path is through the path
$t' \rightarrow a'' \rightarrow z'' \rightarrow t^*$, because $z'$ has already been touched by $p$.
Finally, since $G'$ and $G$ have the same simple paths, 
$p$ cannot lie only within $G'$, otherwise it could be used to show irredundancy of $u \rightarrow v$ in $G$. 
So, $p$ must eventually enter into $G''$ (and then come back into $G'$).

We have two possibilities on how $p$ has first entered into $G'$ through $r'$:
\begin{enumerate}
\item {\em through the edge $r' \rightarrow s'$:} in this case, $p$ cannot reach $G''$ 
through a path of the form $x' \rightarrow r' \rightarrow v''$ otherwise it would be cyclic.
Thus, $p$ has reached $G''$ through a path of the form $u' \rightarrow a' \rightarrow x''$, for some $x'' \in V''$. 
However, $p$ cannot pass through
$u' \rightarrow v'$ before leaving $G'$, otherwise it would be cyclic (since it would touch $u'$ twice).
For the same reason, $p$ cannot pass through $u' \rightarrow v'$ when it comes back into $G'$.
Contradiction.

\item {\em through the edge $r' \rightarrow v''$ and by then entering into $G'$ from $G''$:}
in this case, we have two possible ways for reaching $G'$, i.e. through a path of the form
$u'' \rightarrow a'' \rightarrow x'$ or through a path of the form $x'' \rightarrow r'' \rightarrow v'$.
The first way would force $p$ to touch $a''$ twice (the second one when passing from $t'$ to $t^*$).
But also the second way would make $p$ cyclic because, to include $u' \rightarrow v'$, it should
touch $v'$ twice. Contradiction.
\end{enumerate}

\item {\em Through $r''$:} In this case, $G'$ cannot be reached through $r'' \rightarrow v'$ otherwise,
to include $u' \rightarrow v'$, $p$ should touch $v'$ twice. Thus, it must be that, from $r''$, $p$ follows 
a path of the form $r'' \rightarrow s'' \stackrel{p'}\leadsto u'' \rightarrow a'' \rightarrow x'$, for some $x' \in V'$.
However, having touched $a''$, the only possible way for $p$ to acyclically reach $t^*$ is by coming back
into $G''$, reaching $t''$ and then following the path $t'' \rightarrow a' \rightarrow z' \rightarrow t^*$.
Thus, $p$ cannot come back into $G''$ through a path $u' \rightarrow a' \rightarrow x''$; hence,
the only possibility is via a path of the form $y' \rightarrow r' \rightarrow v''$, for some $y' \in V'$.
Moreover, to be acyclic, $p$ has to move from $v''$ to $t''$ via a path $p''$ that only touches
vertices that $p$ has not touched before.  However, in this way we would have found an $s''t''$- path
$s'' \stackrel{p'}\leadsto u'' \rightarrow v'' \stackrel{p''}\leadsto t''$ that is acyclic and contains the edge
$u'' \rightarrow v''$. This yields a contradiction because, since $G''$ and $G$ have the same simple paths, 
such a path could be used to show irredundancy of $u \rightarrow v$ in $G$.
\end{enumerate}
\end{proof}

\section{Characterisation of Vulnerable Multi-Digraphs}

A characterisation of vulnerable {\em undirected} multigraphs is presented in~\cite{Milch06}. 
In particular, in~\cite{Milch06} it is proved that an undirected multigraph is vulnerable if and only if it 
is not series-parallel; moreover, in~\cite{Milch06} it is also proved that 
this holds if and only if the multigraph does not contain a homeomorphic copy of 
(the undirected version of) $\W$.
In \cite{ChenEtal15},
it is proved that an {\em irredundant} multi-digraph is vulnerable 
if and only if it is not two-terminal series-parallel (TTSP) \cite{RS42}. 

Our main result (Theorem~\ref{thm:vuln} in Section \ref{sec:char}) states that {\em any} $st$-connected 
multi-digraph is vulnerable if and only it contains a homeomorphic copy of $\W$ 
(in the sense made precise by Definition~\ref{def:stEmbedding}).
Therefore, our characterisation extends the result 
in~\cite{Milch06} to multi-digraphs and generalises the result in 
\cite{ChenEtal15} to possibly redundant graphs,
thus answering to Open Question 1 in Sect.\,6.1 of~\cite{Rough06}.
We obtain this result in four steps.

First, in Section~\ref{sec:acyclic}, 
we characterise vulnerability for acyclic multi-digraphs (Fact~\ref{thm:vuln-ser-par}); 
this is an immediate corollary of a result in \cite{ChenEtal15} and standard results 
about two-terminal series-parallel nets. 
Second, in Section~\ref{sec:fromCyclicToAcyclic}, 
we show that two nets with the same set 
of simple $st$-paths have exactly the same 
Wardrop equilibria and, consequently, removing redundant edges 
does not affect vulnerability.
%
Third, in Section \ref{sec:remCycles}, we show that, 
by carefully analysing a simple cycle in a multi-digraph, 
we can always either find  $\W$ as an homeomorphic subgraph 
(again according to Definition~\ref{def:stEmbedding}), 
or find at least a redundant edge.  
Finally, in Section \ref{sec:char}, we apply such cycle analysis 
to characterise vulnerability of general (i.e., possibly cyclic) graphs $G$, 
by either finding a homeomorphic copy of $\W$ in $G$ 
or by finding an 
acyclic graph with the same simple paths as $G$.

\subsection{Characterising Vulnerability for Acyclic Multi-digraphs}
\label{sec:acyclic}

We start by giving the formal definitions of subgraph homeomorphism and $st$-embedding.
In what follows, we always assume that $G$ is $st$-connected. 
This is not restrictive, as we can always find the maximal $st$-connected subgraph of $G$ (through
a visit from $s$ and a backwards visit from $t$), without affecting its vulnerability.
Indeed, as we show in Section \ref{sec:fromCyclicToAcyclic}, edges and vertices that do not 
lie on a simple $st$-path will never play any role in flows and, consequently, they do not influence vulnerability.


\begin{definition}[\cite{LR80}] 
A {\em subgraph homeomorphism} between $H$ and $G$
is a pair of injective mappings $(\phi,\psi)$ such that $\phi : V_H \rightarrow V_G$, 
$\psi : E_H \rightarrow SP(G)$, 
and for every edge $e=x\rightarrow y\in  E_H$, $\psi(e)$ is 
a path in $G$ from $\phi(x)$ to $\phi(y)$.

A homeomorphism is called {\em node-disjoint} if all paths in ${\mathit Image}(\psi)$ are pairwise 
node-disjoint, up-to their end points.
\end{definition}

\begin{definition}
\label{def:stEmbedding}
An {\em st-embedding} of the Wheatstone graph $\W$ (see Fig.~\ref{fig:Wheat}(3)) into a net $(G,s,t)$
is a node-disjoint subgraph homeomorphism $(\phi,\psi)$ between $\W$ and $G$ such that
there are (possibly empty) node-disjoint simple paths from $s$ to $\phi(a)$ and from $\phi(d)$ to $t$ that
are pairwise node-disjoint up-to their endpoints with all paths in ${\mathit Image}(\psi)$.
\end{definition}

By using results in \cite{ChenEtal15} and 
by applying the standard result from \cite{Duf65}, where it is proved that
an acyclic directed graph is TTSP if and only if it does not contain 
a subgraph homeomorphic to $\W$, we can characterise vulnerability for acyclic multi-digraphs,
since acyclicity is a special case of irredundancy.

\begin{fact}
\label{thm:vuln-ser-par}
Let $(G,s,t)$ be an acyclic $st$-connected net. $G$ is vulnerable if and only if
there exists an $st$-embedding of $\W$ into $G$.
\end{fact}

\subsection{From Cyclic to Acyclic Graphs}
\label{sec:fromCyclicToAcyclic}

Here we show soundness of our approach 
by showing that removing edges 
that do not lie on a simple $st$-path 
(i.e., redundant edges) does not affect
vulnerability of the original graph.

Notationally, we write $G' \subseteq G$ whenever $G=(V,E)$ and $G'=(V,E')$, with $E' \subseteq E$;
notation $G' \subset G$ has a similar meaning, with $E' \subset E$.

\begin{lemma}
\label{lemma:equivalentNashFlows}
Let $G'\subseteq G$ and $\varphi$ be a Wardrop flow for $(G,r,l)$. 
If $SP(G')=SP(G)$ then $\varphi$ is a  Wardrop flow for $(G',r,l)$.
\end{lemma}
\begin{proof}
By definition, $\varphi$ is a Wardrop flow if and only if, 
for every $p,q \in P(G)$ such that $\varphi(p) > 0$, it holds that $l_p(\varphi) \leq l_q(\varphi)$. 
By definition, $\varphi$ assigns positive flow only to
acyclic $st$-paths in $G$; thus, $p \in SP(G) = SP(G')$. Moreover, 
since $G'\subseteq G$, it holds that $P(G') \subseteq P(G)$. Thus, trivially, for every $p,q \in P(G')$ such that
$\varphi(p) > 0$, it holds that $l_p(\varphi) \leq l_q(\varphi)$; this means that $\varphi$ 
is a  Wardrop flow for $(G',r,l)$.
\end{proof}

\begin{proposition}
\label{lemma:equivalentGraphs}
Let $G'\subseteq G$.
If $SP(G')=SP(G)$ then $G$ is vulnerable if and only if $G'$ is vulnerable. 
\end{proposition}
\begin{proof}
Vulnerability of $G'$ trivially entails vulnerability of $G$. Let us prove the opposite implication.
Let $H\subset G$ be such that $L(H,r,l|_H)<L(G,r,l)$, for some $r$ and $l$. 
It cannot be $H = G'$, otherwise $SP(H) = SP(G') = SP(G)$ and,
because of Lemma~\ref{lemma:equivalentNashFlows}, 
we would have $L(H,r,l|_H)=L(G,r,l)$. 

If $H \subset G'$, let $\varphi$ be a Wardrop flow for $G$.
By Lemma~\ref{lemma:equivalentNashFlows},
$\varphi$ is a Wardrop flow for $G'$ and hence $L(G,r,l) = L(G',r,l)$; 
thus, $L(H,r,l)<L(G',r,l)$, i.e. $G'$ is vulnerable.

Otherwise, there is an edge $e$ belonging to $H$ but not to $G'$;
this means that $e$ only belongs to cyclic $st$-paths of $G$, because by hypothesis $SP(G')=SP(G)$. 
%
Then, consider $H'$, obtained by removing $e$ from $H$.
Since $P(H) \subseteq P(G)$, we have that $e$ only belongs to cyclic $st$-paths of $H$;
thus, $SP(H')=SP(H)$.
Now, let $\varphi'$ be a Wardrop flow for $H$.
By Lemma~\ref{lemma:equivalentNashFlows},
$\varphi'$ is a Wardrop flow for $H'$ and $L(H,r,l|_H)=L(H',r,l|_{H'})$.
If $H' \subset G'$, then $L(H',r,l|_{H'})<L(G',r,l|_{G'})$ and $G'$ is vulnerable.
Otherwise, we can find another edge to be removed from $H'$; but this
procedure has to terminate, eventually yielding that $G'$ is vulnerable, as desired.
\end{proof}

\newcommand{\I}{\mathsf{E}}
\renewcommand{\O}{\mathsf{X}}
\newcommand{\N}{\mathsf{N}}
\newcommand{\e}{\varepsilon}
\renewcommand{\u}{\xi}
\newcommand{\last}[1]{\hat{#1}}
\newcommand{\fyly}[2]{[#1 \leadsto #2]}
\newcommand{\fyln}[2]{[#1 \leadsto #2)}
\newcommand{\fnly}[2]{(#1 \leadsto #2]}
\newcommand{\fnln}[2]{(#1 \leadsto #2)}
\newcommand{\Nfyly}[3]{[#2 \stackrel {#1} \leadsto #3]}
\newcommand{\Nfyln}[3]{[#2 \stackrel {#1} \leadsto #3)}
\newcommand{\Nfnly}[3]{(#2 \stackrel {#1} \leadsto #3]}
\newcommand{\Nfnln}[3]{(#2 \stackrel {#1} \leadsto #3)}

\subsection{Dealing with Cycles}
\label{sec:remCycles}

In this section, we present our analysis of cycles in a graph.
We distinguish three kinds of cycles.
In all cases, by analysing a cycle $C$ in a net $(G,s,t)$, 
we will come up with one of the following outcomes:
\begin{enumerate*}
\item we find at least one redundant edge in $C$; or
\item we find an $st$-embedding of $\W$ in $G$; or
\item we find a strictly smaller (according to Definition~\ref{def:distance}) cycle $C'$.
\end{enumerate*} 

As usual, we start by giving some preliminary definitions and results.
To this aim, we adopt the following notation: given a path $p \triangleq u_1e_1u_2e_2 \ldots e_n u_{n+1}$, for
$n \geq 0$, we let:

\begin{tabular}{ll}
\hspace*{-.5cm}
$\fyly{u_1}{u_{n+1}} \triangleq u_1e_1u_2e_2 \ldots e_n u_{n+1} (= p)$
\qquad &
$\fyln{u_1}{u_{n+1}} \triangleq u_1e_1u_2e_2 \ldots e_n$
\\
\hspace*{-.5cm}
$\fnly{u_1}{u_{n+1}} \triangleq e_1u_2e_2 \ldots e_n u_{n+1}$
\qquad &
$\fnln{u_1}{u_{n+1}} \triangleq e_1u_2e_2 \ldots e_n$
\end{tabular}

%



\begin{definition}
Let $(G,s,t)$ be a net and let $C$ be a simple cycle in $G$. 

Any simple path of the form $s\leadsto u$ such that $u \in C$ and $C \cap \fyln s u = \emptyset$ 
is called an {\em entry path} (in $C$) and $u$ is said an {\em entry node};
notationally, we shall always denote entry nodes with $\e$. 

Any simple path of the form $u\leadsto t$ such that $u \in C$ and $C \cap \fnly u t = \emptyset$ 
is called an {\em exit path} (from $C$) and $u$ is an {\em exit node};
notationally, we shall always denote exit nodes with $\u$. 
\end{definition}

Clearly, any vertex of a cycle $C$ can be both an entry and an exit for $C$ 
(or it can be neither an entry nor an exit for $C$). 
If $G$ is $st$-connected, then every cycle must have at least one entry and one exit node.
When the cycle has just one entry node or just one exit node, we can easily find at least 
one redundant edge (see Lemma \ref{fact:noteu}). 
Otherwise, we distinguish two kinds of cycles, depending on how entry 
and exit nodes interleave (see Definition \ref{def:splittable}).

\begin{lemma}
\label{fact:noteu}
Let $(G,s,t)$ be an $st$-connected net and $C$ a cycle of $G$
with at most one entry or one exit node.
Then $C$ contains at least one redundant edge.
\end{lemma}
\begin{proof}
Let $\e$ be the only entry node of $C$. The edge of $C$ entering into $\e$ (call it $x \rightarrow \e$) is clearly redundant, because $x$ is reachable from $s$ only via $\e$. Similarly, if $\u$ is the only exit node of $C$, the edge of $C$ exiting from $\u$ (call it $\u\rightarrow x$) is redundant, because $t$ is reachable from $x$ only via $\u$.
\end{proof}


To simplify our analysis of cycles, 
we find it useful to define {\em $s$-minimal cycles} that have the pleasant property 
of having at least an entry path that does not touch any exit path.


\begin{definition}
\label{def:distance}
Let $(G,s,t)$ be a net, $C$ be a simple cycle in $G$ and $d(u,v)$ be the length of the 
shortest path (w.r.t. the number of edges) from node $u$ to node $v$. 

The distance of $C$ from the source is $d_s(C)=\min_{u\in C} d(s,u)$. 
We will denote with $\e^*$ an entry node for $C$ such that $d_s(C)=d(s,\e^*)$.

The distance of $C$ from the target is $d_t(C)=\min_{u\in C} d(u,t)$. 
We will denote with $\u^*$ an exit node for $C$ such that $d_t(C)=d(\u^*,t)$.

We say that $C$ is {\em s-minimal} if, for every cycle $C'$ in $G$, $d_s(C)\leq d_s(C')$.
\end{definition} 

\begin{lemma}
\label{fact:minimal}
Let $(G,s,t)$ be an $st$-connected net, $C$ an $s$-minimal cycle of $G$, 
and $\e^*$ an entry node with minimum distance from $s$. 
Then, there exists a path $s\stackrel p \leadsto \e^*$ such that, for every cycle $C'$ that contains $\e^*$
and for every $\u \leadsto t$ exit path of $C'$, it holds that 
$p \cap \fnly \u t = \emptyset$.
\end{lemma}
\begin{proof}
Let $s \stackrel p \leadsto \e^*$ be a path of minimum length in $G$ entering in $\e^*$, 
$C'$ a cycle that contains $\e^*$ and $\u \stackrel q \leadsto t$ an exit path for $C'$.
If $p$ contains a node $v$ that appears in $q$
($v\not\in\{\u,\e^*\}$, otherwise $p$/$q$ would not be an exit/exit path for $C'$),
then $p$ has the form $s\stackrel{p_1}{\leadsto} v \stackrel{p_2}{\leadsto} \e^*$ 
and $q$  has the form $\u\stackrel{q_1}{\leadsto} v \stackrel{q_2}{\leadsto} t$. 
Then, the cycle  $C''=v \stackrel{p_2}{\leadsto} \e^* \stackrel{C'}{\leadsto} \u \stackrel{q_1}{\leadsto} v$ 
is such that $d_s(C'')\leq d(s,v) < d(s, \e^*) = d_s(C)$, thus contradicting $s$-minimality of $C$.
\end{proof}


\begin{definition}
\label{def:splittable}
Given a simple cycle $C$, we call it {\em splittable} if $C$ is of the form $v_1 \ldots v_k v_1$ and there
exists an $h \in \{1,\ldots,k\}$ such that all entry nodes of $C$ are in $v_1 \ldots v_h$ and all exit nodes
of $C$ are in $v_h \ldots v_k$.
\end{definition}

For a splittable cycle $C$, we can define an order relation among its vertices:
let $\e_1$ be the first entry node in $v_1 \ldots v_h$ and, for every $x,y$ occurring as vertices of $C$ with
$y \neq \e_1$, let $x \leq_C y$ if $C = \e_1 \leadsto x \leadsto y \leadsto \e_1$.
We will denote with $\last{\e}$ (resp. $\last{\u}$) the last entry (resp. exit) node in $C$ (w.r.t. $\leq_C$).
Notice that the definition of splittable cycle allows the last entry coincide with the first exit 
but does {\em not} allow the last exit coincide with the first entry.


\begin{definition}
Given a splittable cycle $C$, 
we call: 
\begin{itemize}
\item
$\I$ the {\em entry region} of $C$, that is the set of nodes in $\Nfyln C {\e_1}{\u_1}$; 
\item
$\O$ the {\em exit region} of $C$, that is the set of nodes in $\Nfyly C {\u_1}{\last\u}$, if $\u_1 \neq \last\e$,
or in $\Nfnly C {\u_1}{\last\u}$, otherwise; 
\item
$\N$ the {\em neutral region} of $C$, that is the set of nodes in $\Nfnln C {\last\u}{\e_1}$.
\end{itemize}
%
We call {\em splitter} the edge of $C$ entering into $\u_1$, if $\u_1\neq \last{\e} $, 
or the vertex $\u_1 (= \last{\e})$, otherwise.
\end{definition}

Notice that, by the above definition, every $v \in \I$ is such that $\e_1 \leq_C v <_C \u_1$
and every $v \in \O$ is such that $\last{\e} <_C v \leq_C \last{\u}$. 
This property is crucial for proving Proposition \ref{prop:hyperChordEntryExit}, that in turn
is a key building block in the proof of Lemma \ref{lemma:eu}.  
To ensure this property,
we have been forced to consider the exit region and the splitter in different ways, according 
to whether $\u_1 \neq \last\e$ or not.

\begin{definition}
A {\em chordal path} for a simple cycle $C$ is a path between two non-adjacent vertices of $C$ that
touches $C$ only in its extremes.

A {\em hyper-chord} for a simple cycle $C$ from $x$ to $y$ ($x,y\in C$) 
is a simple path $x\leadsto y$ that can be decomposed as a 
set of simple paths 
$x=x_1 \leadsto x_2 \leadsto x_3 \leadsto \ldots \leadsto x_{k-1} \leadsto x_k=y$, 
such that: 
\begin{itemize}
\item every $x_i \leadsto x_{i+1}$ is either a chordal path for $C$ or is an edge of $C$;
\item at least one $x_i \leadsto x_{i+1}$ is a chordal path for $C$.
\end{itemize}
%
A hyper-chord $x \stackrel h \leadsto y$
is said {\em neutral} for a splittable cycle $C$ if  $C \cap \Nfnln h x y \subseteq \N$.
\end{definition}

\begin{fact}
\label{propION}
Let $C$ be a splittable simple cycle. Then,
\begin{enumerate}
\item If $p$ is a path from $\I$ to $\O$ whose only vertices of $\I\cup\O$ are its extremes, then either
$p$ touches the splitter or $p$ is a neutral hyper-chord between $\I$ and $\O$.
\item If $p$ is a path from $\I$ to $\N$ whose only vertex of $\I$ is its starting vertex, then either
$p$ touches the splitter or starts with a chordal path going from $\I$ to $\O$ or to $\N$.
\end{enumerate}
\end{fact}

\begin{proposition}
\label{prop:hyperChordEntryExit}
Let $C$ be a splittable cycle and $x \stackrel h \leadsto y$ a neutral hyper-chord from $\I$ to $\O$. 
Then, for every entry or exit path $p$, it holds that $p \cap \Nfnln h x y = \emptyset$.
\end{proposition}
\begin{proof} 
Let us suppose that there exists an internal node $u$ of $h$ that belongs to an exit path. 
If $u$ belongs to a chordal path $w\leadsto z$ with $w\in \I\cup\N$ and $z\in\O\cup\N$, 
then the path $w\leadsto u \leadsto t$ would be an exit path, contradicting the fact 
that $C$ is a splittable cycle (if $w\in \I$)
or the fact that $w\in \N$ (because $\N$ does contain no exit nor entry nodes).
Similarly, if $u$ belongs to an entry path, then the path $s\leadsto u \leadsto z$ would be an entry path, 
contradicting the fact that $C$ is a splittable cycle (if $z\in\O$) or the fact that $z\in \N$. 
\end{proof}

Given a simple splittable cycle $C$, we denote with $f_{\I\N}$ the first (w.r.t. $\leq_C$)
vertex of $\N$ target of a neutral hyper-chord from $\I$; we let $f_{\I\N}$ be $\e_1$ if there
is no neutral hyper-chord from $\I$ to $\N$. Similarly, we denote with $\ell_{\N\O}$ the last (w.r.t. $\leq_C$)
vertex of $\N$ source of a neutral hyper-chord to $\O$; we let $\ell_{\N\O}$ be $\last{\u}$ if there
is no neutral hyper-chord from $\N$ to $\O$. Thus, $\last{\u} <_C f_{\I\N} \leq_C \e_1$ and
$\last{\u} \leq_C \ell_{\N\O} <_C \e_1$.

\begin{lemma}
\label{lemmaNHCIO}
Let $C$ be a splittable simple cycle. 
If $f_{\I\N} \leq_C \ell_{\N\O}$, then there exists a neutral hyper-chord for $C$ from $\I$ to $\O$.
\end{lemma}
\begin{proof}
First, observe that there must exist a neutral hyper-chord $p$ from $\I$ to $\N$, otherwise 
$f_{\I\N} \define \e_1 >_C \ell_{\N\O}$. Similarly, there must exist a neutral hyper-chord $q$ 
from $\N$ to $\O$, otherwise $\ell_{\N\O} \define \last{\u} <_C f_{\I\N}$.  

If $p$ and $q$ intersect, call $a$ the first intersection along $p$ with $q$. Then, consider $p'$, the prefix
of $p$ ending in $a$, and $q'$, the suffix of $q$ starting from $a$. Then, the path $p'\!,\!q'$ is 
a neutral hyper-chord from $\I$ to $\O$.

Consider now the case in which $p$ and $q$ do not intersect; thus, $f_{\I\N} <_C \ell_{\N\O}$.
Thus, we can always find a pair of vertices $x \in p$ and $y \in q$ such that $x <_C y$ and
there exists no other pair $x' \in p$ and $y' \in q$ such that $x \leq_C x' <_C y' \leq_C y$.
Then, consider $p'$, the prefix of $p$ ending in $x$, and $q'$, the suffix of $q$ starting from $y$. 
Then, $p',\! x \stackrel C \leadsto y,\! q'$ is a neutral hyper-chord from $\I$ to $\O$.
\end{proof}

\begin{lemma}
\label{lemmaRedEdges}
If $C$ is a splittable simple cycle without neutral hyper-chords from $\I$ to $\O$,
then all edges of $C$ between $\ell_{\N\O}$ and $f_{\I\N}$ are redundant.
\end{lemma}
\begin{proof}
By Lemma \ref{lemmaNHCIO}, $\ell_{\N\O} <_C f_{\I\N}$ and, by
contradiction, assume that there exists an irredundant edge  
$x \rightarrow y \in \ell_{\N\O} \stackrel C \leadsto f_{\I\N}$. 
This implies that $\ell_{\N\O} \leq_C x$ and $y \leq_C f_{\I\N}$.
Let $\hat \I$ be $\I$, if $\last{\e} \neq \u_1$, and be $\I \cup \{\last{\e}\}$, otherwise.
If the edge $x\rightarrow y$ is irredundant, there should exist a path $p$ from $\hat \I$ to $x$; 
let us call $p'$ the suffix of $p$ 
starting from the last vertex in $\hat \I$ touched by $p$. 
If such a vertex belongs to $\I$, then, by Fact \ref{propION}(2), we only have two possibilities:
\begin{enumerate}
\item[(a)] $p'$ starts with a chordal path from $\I$ to some $x' \in \N$. Then, $p'$ has to reach $x$ but this
cannot be done by touching vertices of $\I \cup \O$ (by construction of $p'$ and by hypothesis); hence,
$p'$ is a neutral hyper-chord from $\I$ to $x$. If $x = \ell_{\N\O} = \last{\u} \in \O$, 
we would contradict the hypothesis that there is no neutral hyper-chord from $\I$ to $\O$. 
Hence $x \in \N$ and, by definition, $f_{\I\N} \leq_C x$; since $x <_C y$, we would contradict $y \leq_C f_{\I\N}$.

\item[(b)] $p'$ touches the splitter. In this case, $p'$ must reach $x$ by jumping at least one exit
node $\u$, that should be used to reach $t$ after touching $x \rightarrow y$. Let $p''$ be the path from
$y$ to $\u$. It cannot touch any vertex in $\I$ otherwise, by Fact \ref{propION}(1), it would touch
the splitter (but then 
$p$ would be cyclic) or would
contain a neutral hyper-chord from $\I$ to $\O$. So, it must be that $p''$ is a neutral hyper-chord from $y$
to $\O$. Like before, if $y = f_{\I\N} = \e_1 \in \I$, 
we would contradict the hypothesis that there is no neutral hyper-chord from $\I$ to $\O$. 
Hence $y \in \N$ and, by definition, $y \leq_C \ell_{\N\O}$; 
since $x <_C y$, this would contradict $\ell_{\N\O} \leq_C x$.
\end{enumerate}
If $\last{\e} = \u_1$, it can also be possible that the last vertex of $\hat \I$ touched by $p$ is the splitter; we reason like in case (b) above.
\end{proof}

Now, we are ready to prove our main result about splittable cycles.

\begin{lemma}
\label{lemma:eu}
Let $G$ be an $st$-connected multidigraph and $C$ be an $s$-minimal splittable cycle in $G$
with at least two entry and two exit nodes. Then either:
\begin{itemize}
\item $C$ contains at least one redundant edge; or
\item $\W$ is $st$-embeddable in $G$; or
\item there exists an $s$-minimal cycle $C'$ such that $d_t(C')<d_t(C)$.
\end{itemize}
\end{lemma}
\begin{proof}
If there is no neutral hyper-chord from $\I$ to $\O$, then, 
by Lemma~\ref{lemmaRedEdges}, 
all edges in $C$ between $\ell_{\N\O}$ and $f_{\I\N}$ are redundant.
Otherwise, if $w\leadsto z$ is a neutral  hyper-chord from $\I$ to $\O$ in $C$, we first observe that 
we have the following situation (left) and hence the following homeomorphic copy of $\W$ in $G$ (right):

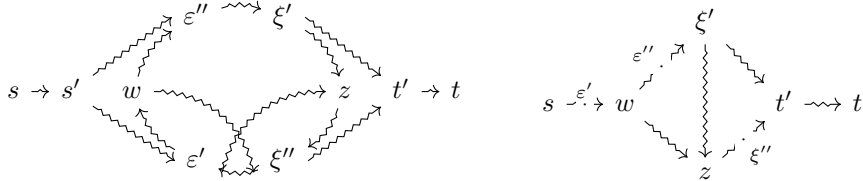
\begin{tabular}{cc}
\hspace*{-.7cm}
\begin{minipage}{0.60\textwidth}
\begin{tikzcd}[column sep=2.5mm,row sep=4mm]
& & & \e'' 
   \ar[rr,rightsquigarrow, bend left=20] 
& & \u' 
   \ar[drr, start anchor = east, rightsquigarrow]
   \ar[dr,rightsquigarrow, bend left=20]
\\
s
\ar[r, rightsquigarrow]
& s'
\ar[urr, end anchor = west, rightsquigarrow] 
\ar[drr, end anchor = west, rightsquigarrow] 
& |[alias=W]|w
\ar[ur,rightsquigarrow, bend left=20]
& & & &
|[alias=Z]|z 
\ar[dl,rightsquigarrow, bend left=20]
&
t'
\ar[r, rightsquigarrow]
& 
t
\\
& & & \e' 
   \ar[ul,rightsquigarrow, bend left=20]
 & & \u'' 
   \ar[urr,start anchor = east, rightsquigarrow]
   \ar[ll,rightsquigarrow, bend left=20, ""{near start, name=H1}, ""{near end, name=H2}] 
\\
\ar[from=W, to=H1, ,rightsquigarrow, bend left=40, start anchor = east, end anchor = north]
\ar[from=H2, to=Z, ,rightsquigarrow, bend left=40, start anchor = north, end anchor = west]
\end{tikzcd}
\end{minipage}
&
\hspace*{-.7cm}
\begin{minipage}{0.35\textwidth}
\begin{tikzcd}[column sep=5mm,row sep=5mm]
& & \u' \ar[dd, rightsquigarrow]\ar[dr, rightsquigarrow] \\
s \ar[r, rightsquigarrow, "\cdot" description, "\e'"] 
& 
w 
\ar[ur, rightsquigarrow, "\cdot" description, "\e''"] 
\ar[dr, rightsquigarrow] & & t' \ar[r, rightsquigarrow]
& 
t
\\
& & z \ar[ur, rightsquigarrow, "\cdot" description, "\u''"'] \\
\end{tikzcd}
\end{minipage}
\vspace*{-.4cm}
\end{tabular}

\noindent
Observe that we can always assume that one between $\u'$ and $\u''$ is $\u^*$,  
the exit node at minimum distance from $t$. Moreover, 
we choose exit paths $\u'\leadsto t$ and $\u''\leadsto t$ among those of minimum length 
and that, once they intersect
for the first time (in $t'$), they become the same path (at most, $t' = t$ and this path is empty).
The same assumption holds for the entry paths $s \leadsto \e'$ and $s \leadsto \e''$ 
(with $s'$ their last intersection), where one between $\e'$ and $\e''$ is $\e^*$.
Finally, remember that a neutral hyper-chord from $\I$ to $\O$  is always disjoint from both entry and exit paths
(see Proposition \ref{prop:hyperChordEntryExit}).

If $s\leadsto \e'$ is disjoint from $\u'\leadsto t$ and $\u''\leadsto t$
(this happens, e.g., if we can let $\e^*$ play the role of $\e'$, i.e., if $\e^* \leq_C w$ -- see Lemma~\ref{fact:minimal}), what we have just shown 
is an $st$-embedding of $\W$. 
Otherwise, $\e'' = \e^*$ and
we distinguish the following scenarios, depending on where the entry path $s'\leadsto \e'$ 
touches $\u'\leadsto t$ and $\u''\leadsto t$ (by Lemma~\ref{fact:minimal}, $s \leadsto s'$ is always disjoint
from every exit path; moreover, the intersections can be only in internal nodes, otherwise the paths would
not be entry/exit paths). We consider the following cases:
\begin{enumerate}
\item $s'\leadsto \e'$ touches the exit paths only in $\fnln {\u'}{t'}$;
\label{case:x1tp}
\item $s'\leadsto \e'$ touches the exit paths only in $\fnln {\u''}{t'}$;
\label{case:x2tp}
\item $s'\leadsto \e'$ touches both exit paths, but only in $\fnln {\u'}{t'}$ and $\fnln {\u''}{t'}$;
\label{case:x12tp}
\item $s'\leadsto \e'$ touches the exit paths in $\fyly {t'}t$.
\label{case:tpt}
\end{enumerate}

\paragraph{Case~\ref{case:x1tp}:}
Let $\omega$ be the {\em last} node 
of $\fnln {\u'}{t'}$ that belongs also to $s' \leadsto \e'$.
By definition of $\omega$, the path $\fnln {\omega}{t'}$ is disjoint from $s\leadsto \omega$ 
and $\omega\leadsto \e'$.
This case is depicted in the following diagram (left) and the corresponding $st$-embedding of $\W$ is depicted on the right.

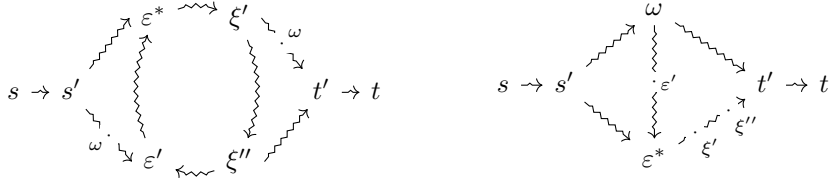
\begin{tabular}{cc}
\hspace*{-.3cm}
\begin{minipage}{0.55\textwidth}
\begin{tikzcd}[column sep=2.5mm,row sep=4mm]
& & & \e^* 
   \ar[rr,rightsquigarrow, bend left=20] 
& & \u' 
   \ar[drr, start anchor = east, rightsquigarrow, "\cdot" description, "\omega"]
   \ar[dd,rightsquigarrow, bend left=20]
\\
s \ar[r, rightsquigarrow]
& s'
\ar[urr, end anchor = west, rightsquigarrow] 
\ar[drr, end anchor = west, rightsquigarrow, "\cdot" description, "\omega"'] 
& & & & &
&
t'
\ar[r, rightsquigarrow]
& 
t
\\
& & & \e' 
   \ar[uu,rightsquigarrow, bend left=20]
 & & \u'' 
   \ar[urr,start anchor = east, rightsquigarrow]
   \ar[ll,rightsquigarrow, bend left=20, ""{near start, name=H1}, ""{near end, name=H2}] 
\\
\end{tikzcd}
\end{minipage}
&
\hspace*{-.7cm}
\begin{minipage}{0.40\textwidth}
\begin{tikzcd}[column sep=3mm,row sep=5mm]
 &&& \omega \ar[dd, rightsquigarrow, "\cdot" description, "\e'"]\ar[drrr, rightsquigarrow] \\
 s \ar[r, rightsquigarrow] & s' \ar[urr, rightsquigarrow] \ar[drr, rightsquigarrow] 
 &&&&& |[alias=T]|t' \ar[r, rightsquigarrow] & t\\
 &&& 
 |[alias=E]|\e^* 
 \ar[urrr, rightsquigarrow, "\cdot"{description, near start, name=U}, "\u'"'{near start}, 
  "\cdot"{description, near end, name=V}, "\u''"'{near end}] 
 \\
\end{tikzcd}
\end{minipage}
\end{tabular}

\paragraph{Case~\ref{case:x2tp}:}
As before, by letting $\omega$ be the {\em last} node 
of $\fnln {\u''}{t'}$ that belongs also to $s'\leadsto \e'$. 
This case is depicted in the following diagram (left) and the corresponding $st$-embedding of $\W$ is depicted on the right.

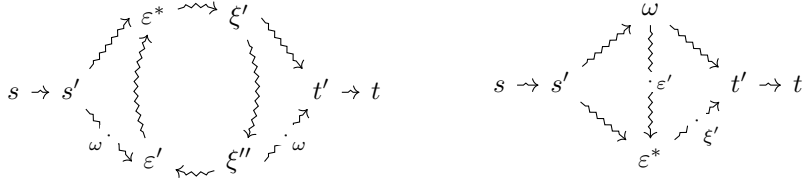
\begin{tabular}{cc}
\hspace*{-.3cm}
\begin{minipage}{0.55\textwidth}
\begin{tikzcd}[column sep=2.5mm,row sep=4mm]
& & & \e^* 
   \ar[rr,rightsquigarrow, bend left=20] 
& & \u' 
   \ar[drr, start anchor = east, rightsquigarrow]
   \ar[dd,rightsquigarrow, bend left=20]
\\
s \ar[r, rightsquigarrow]
& s'
\ar[urr, end anchor = west, rightsquigarrow] 
\ar[drr, end anchor = west, rightsquigarrow, "\cdot" description, "\omega"'] 
& & & & &
&
t'
\ar[r, rightsquigarrow]
& 
t
\\
& & & \e' 
   \ar[uu,rightsquigarrow, bend left=20]
 & & \u'' 
   \ar[urr,start anchor = east, rightsquigarrow, "\cdot" description, "\omega"']
   \ar[ll,rightsquigarrow, bend left=20, ""{near start, name=H1}, ""{near end, name=H2}] 
\\
\end{tikzcd}
\end{minipage}
&
\hspace*{-.75cm}
\begin{minipage}{0.40\textwidth}
\begin{tikzcd}[column sep=3mm,row sep=5mm]
 &&& \omega \ar[dd, rightsquigarrow, "\cdot" description, "\e'"]\ar[drr, rightsquigarrow] \\
 s \ar[r, rightsquigarrow] & s' \ar[urr, rightsquigarrow] \ar[drr, rightsquigarrow] &&&& t' \ar[r, rightsquigarrow] &t\\
 &&& \e^* \ar[urr, rightsquigarrow, "\cdot"{description}, "\u'"'] \\
\end{tikzcd}
\end{minipage}
\end{tabular}

\paragraph{Case~\ref{case:x12tp}:}
Let $\alpha$ and $\beta$ two nodes of $s'\leadsto \e'$ such that $\alpha$ belongs to 
$\fnln {\u'}{t'}$ (resp. $\fnln {\u''}{t'}$) and $\beta$ belongs to $\fnln {\u''}{t'}$ (resp. $\fnln {\u'}{t'}$) 
with no other nodes of $\fnln {\u'}{t'}$ and $\fnln {\u''}{t'}$
in $\alpha\leadsto\beta$. Therefore, we have the $st$-embedding of $\W$ below on the right. 

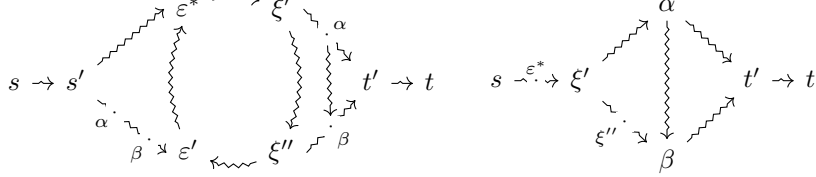
\begin{tabular}{cc}
\hspace*{-.38cm}
\begin{minipage}{0.55\textwidth}
\begin{tikzcd}[column sep=3mm,row sep=4mm]
& & & & \e^* 
   \ar[rr,rightsquigarrow, bend left=20] 
& & \u' 
   \ar[drr, start anchor = east, rightsquigarrow, "\cdot"{description, name=A}, "\alpha"]
   \ar[dd,rightsquigarrow, bend left=20]
\\
s
\ar[r, rightsquigarrow]
& s'
\ar[urrr, end anchor = west, rightsquigarrow] 
\ar[drrr, end anchor = west, rightsquigarrow, "\cdot"{description, near start}, "\alpha"' near start,
"\cdot"{description, near end}, "\beta"' near end] 
& & & & &
& & 
t'
\ar[r, rightsquigarrow]
& 
t
\\
& & & & \e' 
   \ar[uu,rightsquigarrow, bend left=20]
 & & \u'' 
   \ar[urr,start anchor = east, rightsquigarrow,"\cdot"{description, name=B}, "\beta"']
   \ar[ll,rightsquigarrow, bend left=20, ""{near start, name=H1}, ""{near end, name=H2}] 
\\
\ar[from=A, to=B, rightsquigarrow]
\end{tikzcd}
\end{minipage}
&
\hspace*{-.8cm}
\begin{minipage}{0.40\textwidth}
\begin{tikzcd}[column sep=3mm,row sep=5mm]
&&&& \alpha \ar[dd, rightsquigarrow]\ar[drr, rightsquigarrow] \\
s \ar[rr, rightsquigarrow, "\cdot" description, "\e^*"] && \u' \ar[urr, rightsquigarrow] \ar[drr, rightsquigarrow, "\cdot" description, "\u''"'] &&&& t'\ar[r, rightsquigarrow] & t\\
&&&& \beta \ar[urr, rightsquigarrow] \\
\end{tikzcd}
\end{minipage}
\vspace*{-.5cm}
\end{tabular}

\noindent
The case where $\alpha \in \fnln {\u''}{t'}$ and $\beta \in \fnln {\u'}{t'}$ is like the diagrams above,
with the positions of $\alpha$ and $\beta$ swapped.

\paragraph{Case~\ref{case:tpt}:}
Since one between $\u'$ and $\u''$ is $\u^*$, 
let $\omega$ be the last node of $s'\leadsto \e'$ that belongs to $\u^*\leadsto t$ (it is not important 
if $\omega$ is on $\fnln {\u'}{t'}$ or in $t'\leadsto t$).
Consider the cycle $C'=\omega\leadsto \e'\leadsto \e^* \leadsto \u^* \leadsto \omega$ 
(depicted in black in the figures below):

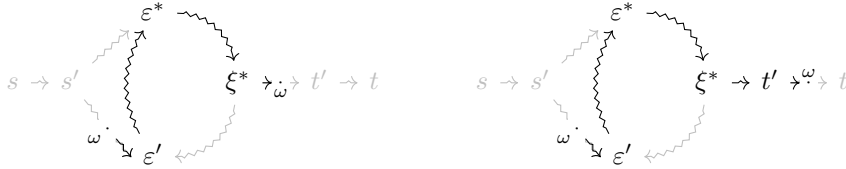
\begin{tabular}{cc}
\hspace*{-.4cm}
\begin{minipage}{0.48\textwidth}
\begin{tikzcd}[column sep=2.5mm,row sep=4mm]
& & & \e^* 
   \ar[drr, rightsquigarrow, bend left=40] 
\\
\textcolor{lightgray}s
  \ar[r, lightgray, rightsquigarrow] & \textcolor{lightgray}{s'}
  \ar[urr, end anchor = south west, rightsquigarrow, lightgray] 
  \ar[drr, end anchor = west, rightsquigarrow, lightgray, "\cdot"{black, description, name=W}, "\omega"'{black}] 
&&&& 
|[alias=U]|\u^* 
  \ar[rr,start anchor = east, rightsquigarrow, lightgray, "\cdot"{black, description, name=Z}, "\omega"'{black}] 
  \ar[dll, rightsquigarrow, lightgray, bend left=40]
&&
\textcolor{lightgray}{t'} \ar[r, rightsquigarrow, lightgray]
& \textcolor{lightgray} t
\\
& & & |[alias=E]|\e' \ar[uu, rightsquigarrow, bend left=30]
\\
\ar[from=W, to=E, rightsquigarrow, end anchor=west]
\ar[from=U, to=Z, rightsquigarrow]
\end{tikzcd}
\end{minipage}
&
\begin{minipage}{0.48\textwidth}
\begin{tikzcd}[column sep=2.5mm,row sep=4mm]
& & & \e^* 
   \ar[drr, rightsquigarrow, bend left=40] 
\\
\textcolor{lightgray}s 
  \ar[r, rightsquigarrow, lightgray] 
& \textcolor{lightgray}{s'}
  \ar[urr, end anchor = south west, rightsquigarrow, lightgray] 
  \ar[drr, end anchor = west, rightsquigarrow, lightgray, "\cdot" {black, description, name=W}, "\omega"'{black}] 
&&&& 
 \u^* 
   \ar[r, rightsquigarrow] 
   \ar[dll, rightsquigarrow, lightgray, bend left=40]
&
|[alias=T]| t'
\ar[rr, rightsquigarrow,  lightgray, "\cdot"{black, description, name=Z}, "\omega"{black}]
& & 
\textcolor{lightgray}t
\\
& & & 
|[alias=E]|\e' 
   \ar[uu, rightsquigarrow, bend left=30]
\\
\ar[from=W, to=E, rightsquigarrow, end anchor=west]
\ar[from=T, to=Z, rightsquigarrow]
\end{tikzcd}
\end{minipage}
\vspace*{-.4cm}
\end{tabular}

\noindent
$C'$ is still $s$-minimal, because $\e^*$ belongs to it. Moreover,
since $\omega$ is on a path of minimum length from $\u^*$ to $t$ 
and $\omega\leadsto t$ does not touch $C'$, 
we have that $d_t(C')\leq d(\omega,t)<d(\u^*, t)=d_t(C)$.
\end{proof}


\begin{lemma}
\label{lemma:eueu}
Let $G$ be an $st$-connected multidigraph and $C$ be an $s$-minimal non-splittable cycle of $G$
with at least two entry and two exit nodes. Then either: 
\begin{itemize}
\item $\W$ is $st$-embeddable in $G$; or
\item there exists an $s$-minimal cycle $C'$ such that $d_t(C')<d_t(C)$.
\end{itemize}
\end{lemma}
\begin{proof}
Since $C$ is $s$-minimal, we can fix one entry node to be $\e^*$ 
and since $C$ is not splittable we can 
always find another entry $\e$ and two 
exits $\u',\u''$ in such a way that $C=\e^*\leadsto \u'\leadsto \e \leadsto \u'' \leadsto \u''$,
with $\u' \neq \e$ 
and one between $\u',\u''$ is $\u^*$ (see picture below (left) -- where we assume the same conventions on $s \leadsto s'$ and $t' \leadsto t$ as in Lemma \ref{lemma:eu}). 
If $s' \leadsto \e$ is disjoint from both exit paths $\u'\leadsto t$ and $\u''\leadsto t$, 
we have the $st$-embedding of $\W$ 
depicted below (right): 

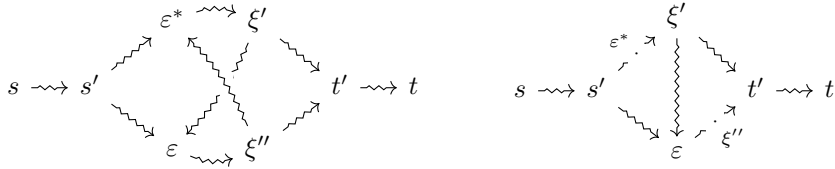
\begin{tabular}{cc}
\hspace*{-.5cm}
\begin{minipage}{0.55\textwidth}
\begin{tikzcd}[column sep=5mm,row sep=3mm]
& & \e^* 
   \ar[r,rightsquigarrow, bend left=20] 
 & \u' 
   \ar[dr,rightsquigarrow]
   \ar[ddl,rightsquigarrow, bend left=10]
\\
s \ar[r,rightsquigarrow] & s' \ar[ur,rightsquigarrow] \ar[dr,rightsquigarrow] 
& & & t'
\ar[r, rightsquigarrow]
& t
\\
& & \e 
   \ar[r, rightsquigarrow, bend right=20] 
 & \u'' 
   \ar[ur,rightsquigarrow]
   \ar[uul,rightsquigarrow, crossing over, bend right=10]
\\
\end{tikzcd}
\end{minipage}
~
\hspace*{-.4cm}
\begin{minipage}{0.4\textwidth}
\begin{tikzcd}[column sep=5mm,row sep=4mm]
&& \u' 
\ar[dd, rightsquigarrow]
\ar[dr, rightsquigarrow] \\
s \ar[r,rightsquigarrow] & s'
\ar[ur, rightsquigarrow, "\e^*", "\cdot" description] 
\ar[dr, rightsquigarrow] & & t' \ar[r,rightsquigarrow] & t\\
&& 
\e \ar[ur, rightsquigarrow, "\cdot" description, "\u''"'] \\
\end{tikzcd}
\end{minipage}
\vspace*{-.2cm}
\end{tabular}

\noindent
Otherwise, we need to distinguish several cases, depending on where $s'\leadsto \e$ 
touches $\u'\leadsto t$ and $\u''\leadsto t$. 
Let $\alpha$ (resp. $\omega$) be the first (resp. last) 
node along $s' \leadsto \e$ that belongs also to an exit path. 
We will consider all possible positions of $\alpha$ and $\omega$ on exit paths: 
\begin{enumerate*}
\item $\alpha\in \fnln {\u'}{t'}$; 
\label{case:alphau1}
\item
$\alpha\in \fnln{\u''}{t'}$; 
\label{case:alphau2}
\item
$\omega\in \fnln {\u'}{t'}$; 
\label{case:omegau2} 
\item
$\omega\in \fyly {t'}t$; 
\label{case:omegatpt} 
\item
$\alpha\in \fyly {t'}t$
\label{case:alphatpt} 
and $\omega\in \fnln {\u''}{t'}$.
\label{case:omegau1} 
\end{enumerate*}

\paragraph{Case~\ref{case:alphau1}:} In the net depicted below (left),
we have the $st$-embedding of $\W$ depicted below (right), 
independently of other intersections between $\alpha \leadsto \e$ and exit paths:
   
\begin{tabular}{cc}
\hspace*{-.6cm}
\begin{minipage}{0.55\textwidth}
\begin{tikzcd}[column sep=5mm,row sep=3mm]
& & \e^* 
   \ar[r,rightsquigarrow, bend left=20] 
 & \u' 
   \ar[dr,rightsquigarrow, "\cdot"{description}, "\alpha"]
   \ar[ddl,rightsquigarrow, bend left=10]
\\
s \ar[r, rightsquigarrow]
& s' \ar[ur,rightsquigarrow] 
   \ar[dr,rightsquigarrow, "\cdot"{description}, "\alpha"'] 
& & & t'
\ar[r, rightsquigarrow]
& t
\\
& & \e 
   \ar[r,rightsquigarrow, bend right=20] 
 & \u'' 
   \ar[ur,rightsquigarrow]
   \ar[uul,rightsquigarrow, crossing over, bend right=10]
\\
\end{tikzcd}
\end{minipage}
~
\hspace*{-.6cm}
\begin{minipage}{0.4\textwidth}
\begin{tikzcd}[column sep=4mm,row sep=4mm]
&&& \u' 
\ar[dd, rightsquigarrow]
\ar[drr, rightsquigarrow, "\cdot"{description, near start}, "\e" near start,
       "\cdot"{description, near end}, "\u''" near end] \\
s \ar[r, rightsquigarrow] & s' 
\ar[urr, rightsquigarrow, "\cdot" description, "\e^*"] 
\ar[drr, rightsquigarrow] &&&&  t' \ar[r, rightsquigarrow] & t\\
&&&
\alpha 
\ar[urr, rightsquigarrow] \\
\end{tikzcd}
\end{minipage}
\end{tabular}

\paragraph{Case~\ref{case:alphau2}:}
In the net depicted below (left), we have the $st$-embedding 
of $\W$ depicted below (right), independently of other intersections between $\alpha\leadsto \e$ and exit paths:

\begin{tabular}{cc}
\hspace*{-.6cm}
\begin{minipage}{0.55\textwidth}
\begin{tikzcd}[column sep=5mm,row sep=3mm]
& & \e^* 
   \ar[r,rightsquigarrow, bend left=20] 
 & \u' 
   \ar[dr,rightsquigarrow]
   \ar[ddl,rightsquigarrow, bend left=10]
\\
s \ar[r, rightsquigarrow]& s' \ar[ur,rightsquigarrow] 
   \ar[dr,rightsquigarrow, "\cdot"{description}, "\alpha"'] 
& & & t'
\ar[r, rightsquigarrow]
& t
\\
& & \e 
   \ar[r,rightsquigarrow, bend right=20] 
 & \u'' 
   \ar[ur, rightsquigarrow, "\cdot"{description}, "\alpha"']
   \ar[uul,rightsquigarrow, crossing over, bend right=10]
\\
\end{tikzcd}
\end{minipage}
~
\hspace*{-.6cm}
\begin{minipage}{0.4\textwidth}
\begin{tikzcd}[column sep=6mm,row sep=4mm]
&& \u' 
\ar[dd, rightsquigarrow, "\cdot"{description, near start}, "\e" near start,
       "\cdot"{description, near end}, "\u''" near end]
\ar[dr, rightsquigarrow] \\
s \ar[r, rightsquigarrow]
& s' 
\ar[ur, rightsquigarrow, "\cdot" description, "\e^*"] 
\ar[dr, rightsquigarrow] & &  t' \ar[r, rightsquigarrow]
& t
\\
&& 
\alpha 
\ar[ur, rightsquigarrow] \\
\end{tikzcd}
\end{minipage}
\end{tabular}

\paragraph{Case~\ref{case:omegau2}:}
In the net depicted below (left), we have the $st$-embedding 
of $\W$ depicted below (right), independently of other intersections between $s'\leadsto \omega$ and exit paths: 

\begin{tabular}{cc}
\hspace*{-.6cm}
\begin{minipage}{0.55\textwidth}
\begin{tikzcd}[column sep=5mm,row sep=3mm]
& & \e^* 
   \ar[r,rightsquigarrow, bend left=20] 
 & \u' 
   \ar[dr,rightsquigarrow, "\cdot"{description}, "\omega"]
   \ar[ddl,rightsquigarrow, bend left=10]
\\
s \ar[r, rightsquigarrow]
&
s' \ar[ur,rightsquigarrow] 
   \ar[dr,rightsquigarrow, "\cdot"{description}, "\omega"'] 
& & & t'
\ar[r, rightsquigarrow]
& t
\\
& & \e 
   \ar[r,rightsquigarrow, bend right=20] 
 & \u'' 
   \ar[ur, rightsquigarrow]
   \ar[uul,rightsquigarrow, crossing over, bend right=10]
\\
\end{tikzcd}
\end{minipage}
~
\hspace*{-.6cm}
\begin{minipage}{0.4\textwidth}
\begin{tikzcd}[column sep=5mm,row sep=4mm]
& & \omega 
\ar[dd, rightsquigarrow]
\ar[dr, rightsquigarrow] \\
s 
\ar[r, rightsquigarrow, "\cdot"{description}, "\e^*"]
&
\u' 
\ar[ur, rightsquigarrow] 
\ar[dr, rightsquigarrow] & &  t' \ar[r, rightsquigarrow]
& t
\\
& & 
\e 
\ar[ur, rightsquigarrow, "\cdot" description, "\u''"'] \\
\end{tikzcd}
\end{minipage}
\end{tabular}

\paragraph{Case~\ref{case:omegatpt}:}
In the net depicted below, we can consider the cycle $C'=\omega\leadsto\e\leadsto \u''\leadsto \e^*
\leadsto \u'\leadsto t'\leadsto \omega$ (in black): 

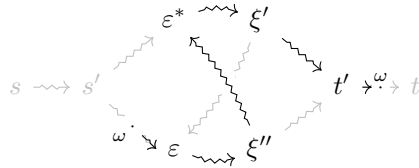
\begin{center}
\begin{minipage}{0.5\textwidth}
\begin{tikzcd}[column sep=5mm,row sep=3mm]
& & \e^* 
   \ar[r, rightsquigarrow, bend left=20] 
 & \u' 
   \ar[dr, rightsquigarrow]
   \ar[ddl,rightsquigarrow, lightgray, bend left=10]
\\
\textcolor{lightgray}s  \ar[r, rightsquigarrow, lightgray]
& |[alias=S]| \textcolor{lightgray}{s'} \ar[ur,rightsquigarrow, lightgray] 
   \ar[dr,rightsquigarrow, lightgray, end anchor=west, "\cdot"{black, description,name=O}, "\omega"'{black}] 
& & & |[alias=P]|t'
\ar[r, rightsquigarrow, lightgray, "\cdot"{black, description, name=Q}, "\omega"{black}]
& |[alias=T]| \textcolor{lightgray}t
\\
& & |[alias=E]|\e 
   \ar[r, rightsquigarrow, bend right=20] 
 & \u'' 
   \ar[ur, rightsquigarrow, lightgray]
   \ar[uul, rightsquigarrow, crossing over, bend right=10]
\\
\ar[from=O, to=E, rightsquigarrow, end anchor=west]
\ar[from=P, to=Q, rightsquigarrow]
\end{tikzcd}
\end{minipage}
\vspace*{-.6cm}
\end{center}

\noindent
Since $C'$ contains $\e^*$, it is still an $s$-minimal cycle; moreover, $d_t(C')\leq d(\omega, t) < \min\{d(\u',t),d(\u'',t)\}=d_t(C)$, since $\u'\leadsto t$ and 
$\u''\leadsto t$ are paths of minimum length and $\omega\leadsto t$ does not touch $C'$.

\paragraph{Case~\ref{case:alphatpt}:} 
If we are not in one of cases \ref{case:alphau1}--\ref{case:omegatpt}, then the situation is depicted below:

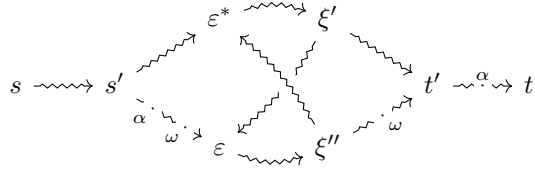
\begin{center}
\begin{tikzcd}[column sep=8mm,row sep=3mm]
& & \e^* 
   \ar[r,rightsquigarrow, bend left=20] 
 & \u' 
   \ar[dr,rightsquigarrow]
   \ar[ddl,rightsquigarrow, bend left=10]
\\
s \ar[r,rightsquigarrow] &
s' \ar[ur,rightsquigarrow] 
   \ar[dr,rightsquigarrow, "\cdot"{description, near end}, "\omega"' near end, 
  						   "\cdot"{description, near start}, "\alpha"' near start ] 
& & & t'
\ar[r,rightsquigarrow, "\cdot"{description}, "\alpha"]
& t
\\
& & \e 
   \ar[r,rightsquigarrow, bend right=20] 
 & \u'' 
   \ar[ur, rightsquigarrow, "\cdot"{description, name=W}, "\omega"']
   \ar[uul,rightsquigarrow, crossing over, bend right=10]
\\
\end{tikzcd}
\vspace*{-.4cm}
\end{center}

\noindent
If $\alpha\leadsto\omega$ does not touch $\fnln {\u'}{t'}$,
let $t^*$ be the first node  along $t'\leadsto t$
that appears in $s'\leadsto \e$ (of course, $t^*$ can be $\alpha$ or some node between $\alpha$ and $\omega$ 
along $s' \leadsto \e$). 
We can then consider the cycle $C'=t^*\leadsto\omega\leadsto\e\leadsto \u''\leadsto \e^*
\leadsto \u'\leadsto t'\leadsto t^*$ (that has the same shape of the cycle considered in Case~\ref{case:omegatpt}). 
$C'$ is $s$-minimal (since it contains $\e^*$) and 
$d_t(C')\leq d(t^*,t) < \min\{d(u',t),d(u'',t)\} = d_t(C)$.

Otherwise, $\alpha\leadsto\omega$ touches $\fnln {\u'}{t'}$
in a set of vertices $B=\{\beta_1, \ldots, \beta_k\}$ and
we must have a path from some $\beta \in B$ to some $\gamma$ in $\fnln {\u''}{t'}$ that has no further
intersections with $\u'\leadsto t'$ and $\u''\leadsto t'$. 
If this path does not touch $t'\leadsto t$ (see picture below, left), 
we can 
find an $st$-embedding of $\W$ (see picture below, right).

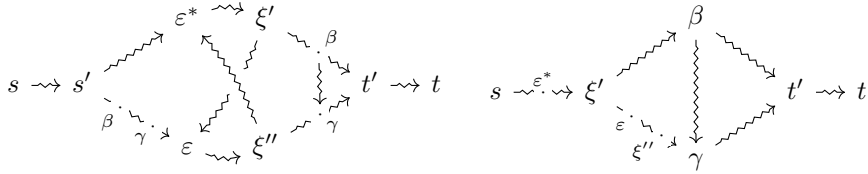
\begin{tabular}{cc}
\hspace*{-.7cm}
\begin{minipage}{0.5\textwidth}
\begin{tikzcd}[column sep=4mm,row sep=3mm]
&&& \e^* 
   \ar[r,rightsquigarrow, bend left=20] 
 & \u' 
   \ar[drr,rightsquigarrow, "\cdot"{description, name=B}, "\beta"]
   \ar[ddl,rightsquigarrow, bend left=10]
\\
s \ar[r, rightsquigarrow]
& s'
\ar[urr,rightsquigarrow] 
   \ar[drr,rightsquigarrow, "\cdot"{description, near end}, "\gamma"' near end, 
  						   "\cdot"{description, near start}, "\beta"' near start ] 
&&&&& t'
\ar[r, rightsquigarrow]
& t
\\
&&& \e 
   \ar[r,rightsquigarrow, bend right=20] 
 & \u'' 
   \ar[urr, rightsquigarrow, "\cdot"{description, name=W}, "\gamma"']
   \ar[uul,rightsquigarrow, crossing over, bend right=10]
\\
\ar[from=B, to=W, rightsquigarrow]
\end{tikzcd}
\end{minipage}
~
\begin{minipage}{0.4\textwidth}
\begin{tikzcd}[column sep=4mm,row sep=4mm]
&&&& \beta 
\ar[dd, rightsquigarrow]
\ar[drr, rightsquigarrow] \\
s
\ar[rr, rightsquigarrow, "\cdot"{description}, "\e^*"]
&&
\u' 
\ar[urr, rightsquigarrow] 
\ar[drr, rightsquigarrow, "\cdot"{description, near start}, "\e"' near start, "\cdot"{description, near end}, "\u''"' near end]
&&&&  t'\ar[r, rightsquigarrow]
& t
\\
&&&&
\gamma 
\ar[urr, rightsquigarrow] \\
\end{tikzcd}
\end{minipage}
\vspace*{-.2cm}
\end{tabular}

Otherwise, let $\beta^*$ be the first vertex of $B$ along $\fnln {\u'}{t'}$ and 
$t^*$ the first vertex along  $\fyly {t'}{t}$ after $\beta^*$ 
in $\fnln {s'} \e$:

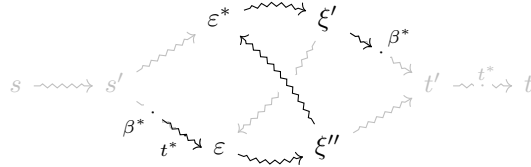
\begin{center}
\begin{tikzcd}[column sep=8mm,row sep=3mm]
& & \e^* 
   \ar[r,rightsquigarrow, bend left=20] 
 & |[alias=UP]|\u' 
   \ar[dr,rightsquigarrow, lightgray, "\cdot"{black, description, name=B}, "\beta^*"{black}]
   \ar[ddl,rightsquigarrow, lightgray, bend left=10]
\\
\textcolor{lightgray}s \ar[r,rightsquigarrow, lightgray] &
\textcolor{lightgray}{s'} \ar[ur,rightsquigarrow,lightgray] 
   \ar[dr,rightsquigarrow, end anchor = west, lightgray,  
        "\cdot"{black, description, near end, name=TP}, "t^*"'{black, near end}, 
  		"\cdot"{black, description, near start, name=BB}, "\beta^*"'{black, near start} ] 
& & & \textcolor{lightgray}{t'}
\ar[r,rightsquigarrow, lightgray, "\cdot"{lightgray, description, name=T}, "\ t^*"{lightgray}]
& \textcolor{lightgray}t
\\
& & |[alias=E]|\e 
   \ar[r,rightsquigarrow, bend right=20] 
 & \u'' 
   \ar[ur, rightsquigarrow, lightgray]
   \ar[uul,rightsquigarrow, crossing over, bend right=10]
\\
\ar[from=BB, to=E, rightsquigarrow, end anchor=west]
\ar[from=UP, to=B, rightsquigarrow]
\end{tikzcd}
\vspace*{-.6cm}
\end{center}

\noindent
Then, the path $C'=\beta^* \leadsto t^* \leadsto\e\leadsto \u''\leadsto \e^*
\leadsto \u'\leadsto \beta^*$ (in black above) is an $s$-minimal cycle 
(since it contains $\e^*$) such that $d_t(C')\leq d(t^*,t) < \min\{d(u',t),d(u'',t)\} = d_t(C)$.
\end{proof}

\subsection{Characterising Vulnerability}
\label{sec:char}

We are now ready for giving the characterisation of vulnerability for general $st$-multi-digraphs.

\begin{theorem}
\label{thm:vuln}
Let $(G,s,t)$ be an $st$-connected net. 
$G$ is vulnerable if and only if there exists an $st$-embedding of $\W$ into $G$. 
\end{theorem}
\begin{proof}
For the `if' part,
we know that $G$ admits a subgraph of the form
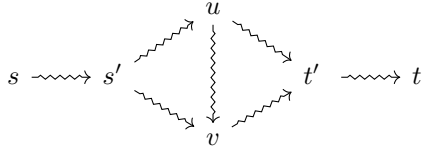
\begin{center}
\begin{tikzcd}[column sep=.8cm,row sep=.4cm]
 & & u \ar[dd, rightsquigarrow]\ar[dr, rightsquigarrow]\\
s \ar[r, rightsquigarrow] & s' \ar[ur, rightsquigarrow]\ar[dr, rightsquigarrow] && t'\ \ar[r, rightsquigarrow] & t\\
  & & v \ar[ur, rightsquigarrow] & \\
\end{tikzcd}
\vspace*{-.5cm}
\end{center}
where all paths are node-disjoint, except for the extremal nodes. 
Let us consider the latency assignment $l$ that assigns:
\begin{itemize}
\item 0 to all edges in $u \leadsto v$, in $s \leadsto s'$ and in $t' \leadsto t$;
\item $x$ to the first edge in $s' \leadsto u$ and in $v \leadsto t'$ and 0 to all the remaining edges in those paths;
\item 1 to the first edge in $s' \leadsto v$ and in $u \leadsto t'$ and 0 to all the remaining edges in those paths;
\item $\infty$ to all the remaining edges. 
\end{itemize}
In this way, $G$ behaves like the Wheatstone net; thus, it is vulnerable.

For the `only if' part, 
if $G$ is not cyclic, the statement follows by Fact~\ref{thm:vuln-ser-par}.
If $G$ contains cycles, by contradiction, let us suppose that there exists no $st$-embedding of $\W$ into it. 
Then, take an $s$-minimal cycle and, according to its kind, apply one of
Lemma~\ref{fact:noteu}, \ref{lemma:eu} or \ref{lemma:eueu}. 
Such results can lead either to deleting redundant edges or to change the cycle.
However, since the new cycle is strictly closer to $t$ than the old one, we must eventually
delete some redundant edge(s). By repeating this reasoning, 
we will eventually transform $G$ into an acyclic subgraph $G'$ such that 
$SP(G')=SP(G)$. By Proposition \ref{lemma:equivalentGraphs}, since $G$ is vulnerable, $G'$ is vulnerable too; thus, since $G'$ is acyclic, there exists an $st$-embedding of $\W$ into $G'$ (by Fact~\ref{thm:vuln-ser-par}) that 
is also an $st$-embedding of $\W$ into $G$. Contradiction.
\end{proof}

\section{Consequences of the Characterisation}

\subsection{On Polynomially Checking Vulnerability}
\label{sec:polyAlg}

Stemming from the characterisation of vulnerable nets provided by Theorem \ref{thm:vuln},  
we define function \fun{isVulnerable} (its pseudocode is given in Algorithm \ref{alg})
that detects if $\W$ 
is $st$-embeddable in a given net $(G,s,t)$. It essentially implements the procedure hinted to
in the `only if' part of the proof of Theorem \ref{thm:vuln}.

\subsubsection{Algorithm Description}
The outer {\bf while} loop of function  \fun{isVulnerable} (lines \ref{alg:whileStart}--\ref{alg:whileEnd}) 
runs until $G$ is acyclic 
or an $st$-embedding of $\W$ is found (line \ref{alg:returnW}).  
This loop starts by computing an $s$-minimal cycle $C$ of $G$ (line \ref{alg:sMinimal}). 

The inner {\bf repeat} loop (lines \ref{alg:repeatStart}--\ref{alg:repeatEnd}) 
runs until cycle analysis in its body 
either finds at least a redundant edge to delete or finds an $st$--embedding of $\W$. 
This loop starts by computing the set of entry (line \ref{alg:entryNodes}) 
and exit (line \ref{alg:exitNodes})  
nodes of $C$ and then it determines which kind of cycle $C$ is (lines \ref{alg:oneEUStart}--\ref{alg:endif}). 

If $C$ has just one entry node (line \ref{alg:oneEntry}) 
or just one exit node (line \ref{alg:oneExit}), 
according to Lemma \ref{fact:noteu},  
a redundant edge is found and it is stored in the set $D$ of edges that will be deleted from $G$. 

If $C$ is a splittable cycle, according to Lemma \ref{lemma:eu}, function \fun{splittableAnalysis} (line \ref{alg:splittableAnalysis})  
returns $\langle \fun{true}, \varnothing, C\rangle$ if 
it finds an $st$-embedding of $\W$ in $G$, $\langle \fun{false}, D, C \rangle$ if it identifies a set $D\not=\varnothing$ of redundant edges to be deleted from $G$, and $\langle \fun{false}, \varnothing, C' \rangle$ otherwise, where $C'$ is a cycle closer to $t$ than $C$. 

If $C$ is not splittable, according to Lemma \ref{lemma:eueu}, function \fun{nonSplittableAnalysis} (line \ref{alg:nonSplittableAnalysis}) returns $\langle \fun{true}, C \rangle$ if it finds a $st$-embedding of $\W$ and 
$\langle \fun{false}, C'\rangle$ otherwise, where $C'$ is a cycle closer to $t$ than $C$.

At the beginning of the function (line \ref{alg:makeST1}) and after each edge deletion (line \ref{alg:makeST2}), 
we invoke function \fun{makeSTConnected} in order to ensure the invariant that $G$ is an $st$-connected net. 
 
At the end, if the {\bf while} loop terminates because $G$ is acyclic, we just verify if $G$ is two-terminal 
series--parallel or not (line \ref{alg:ttsp}).

\begin{algorithm}[t]
  \caption[vulnerable or not vulnerable]
  {Checking Vulnerablity}
  \label{alg}
  \begin{algorithmic}[1]
    \REQUIRE
    {A directed multigraph $G = (V,E)$, $s$ is the source and $t$ is the target}
    \ENSURE {\fun{isVulnerable}$(G,s,t)$} 

    \STATE $G$ $=$ \fun{makeSTConnected}($G,s,t$)
    \label{alg:makeST1}
    \WHILE{$G$ is cyclic}
    \label{alg:whileStart}
      \STATE $\langle C,\e^*\rangle$ $=$ \fun{s-minimalCycle}($G,s,t$)
      \label{alg:sMinimal}
      \REPEAT
          \label{alg:repeatStart}
      \STATE $C'=C$; $Vuln=\fun{false}$
      \STATE $En$ $=$ \fun{entryNodes}($G,C,s$)
      \label{alg:entryNodes}
      \STATE $Ex$ $=$ \fun{exitNodes}($G,C,t$)
      \label{alg:exitNodes}
      \IF{$En=\{\e\}$} 
      \label{alg:oneEUStart}
         \STATE $D=\{\text{the edge of $C$ that enters into }\e\}$
         \label{alg:oneEntry}
      \ELSIF{$Ex=\{\u\}$}   
         \STATE $D=\{\text{the edge of $C$ that exits from }\u\}$
         \label{alg:oneExit}
	  \label{alg:oneEUEnd}
      \ELSIF{\fun{isSplittable}($C, En, Ex,G$)} 
      \label{alg:isSplittable}  
         \STATE $\langle Vuln, C, D\rangle$ $=$ \fun{splittableAnalysis}($C,En,Ex,G$) 
         \label{alg:splittableAnalysis}
      \ELSE 
      	 \STATE $\langle Vuln, C \rangle$ $=$ \fun{nonSplittableAnalysis}($C,En,Ex,G$) 
	  \label{alg:nonSplittableAnalysis}  
	  \label{alg:endif}  
	  \ENDIF
	  \UNTIL{$C\not=C'$}    
	            \label{alg:repeatEnd}
	  \IF{$Vuln$}
	     \STATE {\bf return} $true$
	     \label{alg:returnW}
	  \ELSE
		 \STATE $G$ $=$ \fun{makeSTConnected}($(V, E \setminus D), s, t$)
    	\label{alg:whileEnd}
	    \label{alg:makeST2}
	  \ENDIF	 
    \ENDWHILE 
    \STATE {{\bf return} $\neg${\sc TTSP($G$)}}
    \label{alg:ttsp}

\end{algorithmic}
\end{algorithm}  

\subsubsection{Complexity Analysis} 

We start by describing how functions \fun{entryNodes}, \fun{exitNodes}, \fun{splittableAnalysis} and \fun{nonSplittableAnalysis} 
can be implemented and by analysing their complexity. Then, we analyse the overall complexity of function 
\fun{isVulnerable}. 

As it should be clear from the proofs of 
Lemma~\ref{lemma:eu} and~\ref{lemma:eueu}, 
we only need to solve shortest path problems, reachability problems and to detect intersections between paths.
Moreover, since we are dealing with $st$-connected multigraphs, the number of
edges $|E|$ is at least $|V|-1$. 

In what follows, to make notation lighter, given $X \subseteq V$,
we denote by $in(X)$ (resp, $out(X)$) the set of edges $\{in(x) \ :\ x \in X\}$ (resp, $\{out(x) \ :\ x \in X\}$).

\begin{fact}
\label{lemma:entryExitNodes}
Functions \fun{entryNodes} and \fun{exitNodes} cost $O(|E|)$. Furthermore, 
they can associate: (1) a minimum length entry/exit path to every entry/exit node;
and, (2) the minimum distance from $s$ and from $t$ to every node in such paths.
\end{fact}
\begin{proof}
Function \fun{entryNodes} 
can be implemented via a BFS from $s$ in $(V , E \setminus out(C))$;
entry nodes of $C$ are those with finite distance from $s$. 
Dually, \fun{exitNodes} computes exit nodes by a backwards BFS from $t$
in $(V , E \setminus in(C))$.
By construction, each BFS builds a minimum spanning tree that can be used to
find a minimum length path for every
visited node; the minimum distance is the height of the node in the tree. This is all done in $O(|E|)$.
\end{proof}

\begin{remark}
\label{rem:oV}
We will store the four entry/exit paths needed for the cycle analysis as an array of nodes 
indexed by the distance of the node from $s$ (for entry paths) or $t$ (for exit paths). 
These can be efficiently built from the spanning trees produced by the BFSs of Fact \ref{lemma:entryExitNodes},
In this way, intersections between entry/exit paths within 
functions \fun{splittableAnalysis} and \fun{nonSplittableAnalysis}
can be computed in $O(|V|)$.
E.g., if we need to find the first (resp. last) intersection along an entry path $p$
with an exit path $q$, it suffices to scan from left to right (resp., from right to left)
the array storing $p$: since we are dealing with shortest paths, 
if the $i$-th node of $p$ has distance $k$ from $t$, 
it can only occur in the $k$-th position of the array representing $q$. A dual reasoning is
needed if $p$ is an exit and $q$ is an entry path.
\end{remark}

\begin{lemma}
\label{lemma:saComplexity}
Function \fun{splittableAnalysis} costs $O(|E|)$.
\end{lemma}
\begin{proof}
This function implements the case analysis described in the proof of Lemma \ref{lemma:eu}. 
It starts with calculating entry ($\I$), exit ($\O$), and neutral ($\N$) regions of $C$;
these are computed just by a scansion of $C$ in $O(|V|)$. 

The next step is to search for a neutral hyper-chord from $\I$ to $\O$. 
This can be done by a backwards BFS in $(V, E \setminus (out(\I) \cup in(\O)))$ 
that starts from the set of nodes $\O$ and costs $O(|E|)$. 
This search computes the set $W\subseteq\I$ of nodes that are the source of 
a neutral hyper-chord from $\I$ to $\O$. 

If no such hyper-chord exists ($W=\varnothing$), we have to compute
$f_{\I\N}$, $\ell_{\N\O}$, and all edges of $C$ between $\ell_{\N\O}$ and $f_{\I\N}$
(that are redundant). This can be done by two BFSs (the first starts from all $\I$
and is run in $(V, E \setminus (in(\I) \cup in(\O)))$, the second starts from all $\N$
and is run in $(V, E \setminus (in(\I) \cup out(\O)))$), followed by a linear scansion of $C$. 
The overall cost is again  $O(|E|)$.

Otherwise, we can choose any element of $W\not=\varnothing$ as $w$. 
However, to speed-up convergence of the algorithm, 
it is convenient to choose $w$ as the maximum element of $W$ (w.r.t. $\leq_C$). 
%
If $\e^*\leq_C w$ (this can be checked in $O(|V|)$), we return $\langle\fun{true},\varnothing, C\rangle$. 
Otherwise, we let $\e''$ be $\e^*$ and consider the first entry node as $\e'$ (this is a simplifying choice:
every $\e' \leq_C w$ would work); then, we choose $\u'$ and $\u''$ so that one of them is $\u^*$.
 
If the entry path in $\e'$ does not intersect exit paths, we return $\langle\fun{true}, \varnothing, C\rangle$. 
Otherwise, we check in which case (among \ref{case:x1tp}, \ref{case:x2tp}, \ref{case:x12tp}, and ~\ref{case:tpt} described in the proof of Lemma \ref{lemma:eu}) we are. 
By Remark \ref{rem:oV}, this can be done in $O(|V|)$.
In this check, also $\omega$ (cases \ref{case:x1tp}, \ref{case:x2tp}, and \ref{case:tpt}) and $\alpha/\beta$ (case \ref{case:x12tp}) can be identified. 
According to the case in which we fall, we return either  $\langle\fun{true},\varnothing, C\rangle$
or  $\langle\fun{false},\varnothing, C'\rangle$, where $d_t(C')<d_t(C)$, 
without any additional computational cost.

Summing up, the overall complexity is $O(|E|)$.
\end{proof}

\begin{lemma}
\label{lemma:nsaComplexity}
Function \fun{nonSplittableAnalysis} costs $O(|V|)$.
\end{lemma}
\begin{proof}
Function  \fun{nonSplittableAnalysis} implements the case analysis 
for non-splittable cycles described in the proof of Lemma \ref{lemma:eueu}.
We pay an $O(|V|)$ for choosing $\e,\e^*, \u', \u''$. Then, we need to check in which
of the possible situations put forward by the Lemma we fall: 
this only requires to list all intersections between 
two minimum paths, that, by Remark \ref{rem:oV}, costs $O(|V|)$.
\end{proof}

\begin{theorem}
Checking whether a net $(G,s,t)$ is vulnerable 
can be solved in time polynomial to the size of $G$, in particular its complexity is $O(|V|\cdot|E|^2)$.
\end{theorem}
\begin{proof}
Correctness of function \fun{isVulnerable} in Algorithm 1
stems from the proof of Theorem \ref{thm:vuln}. 


As for the inner {\bf repeat} loop, 
at every iteration $d_t(C)$ decreases by at least 1 
and thus, this loop terminates after at most $O(|V|)$ iterations. 
By Fact \ref{lemma:entryExitNodes} and Lemma \ref{lemma:saComplexity} and \ref{lemma:nsaComplexity}, 
the body of this loop costs $O(|E|)$.
Indeed, lines \ref{alg:oneEUStart}--\ref{alg:oneEUEnd} correspond to Lemma \ref{fact:noteu}, whereas
function \fun{isSplittable} 
checks whether all entry nodes come before all exit nodes in $C$; both these tasks can be done via a
scansion of $C$ and thus cost $O(|V|)$.

As for the outer {\bf while} loop, 
at each iteration the number of edges of $G$ strictly decreases by at least 1;
thus, this loop terminates after at most $O(|E|)$ iterations. 
The complexity of the body of this loop is dominated by 
the cost of the inner {\bf repeat} loop, that, by the above considerations, is  $O(|V|\cdot|E|)$. 
Indeed, making an $st$-multidigraph $st$-connected (function \fun{makeSTConnected} invoked in line 
\ref{alg:makeST1} and \ref{alg:makeST2})
consists of a DFS from $s$ and a backwards DFS from $t$;
all nodes not visited in {\em both} DFSs can be deleted together with their incident edges. 
This costs $O(|E|)$.
An $s$-minimal cycle $C$ and the corresponding $\e^*$ 
(function \fun{s-minimalCycle} in line \ref{alg:sMinimal})
can be found in $O(|V|\cdot|E|)$. 
This covers also the complexity of the guard of the {\bf while} loop in line \ref{alg:whileStart}. 
Indeed, to find $C$ and $\e^*$, consider every successor $v$ of $s$ and see whether $v$ belongs to
some cycle (this can be easily done by a DFS in $G$ starting from $v$). 
If so, $v$ is $\e^*$ and the cycle found is $C$.
Otherwise, iterate the reasoning with the successors of the successors of $s$; and so on until a cycle
is found or all vertices have been considered (in this case $G$ is acyclic).

%
Finally, checking if an acyclic graph is two-terminal series-parallel can be executed in linear time (see \cite{Schoenmakers95anew,VTL82}).

To sum up, the overall complexity of our algorithm is $O(|V| \cdot |E|^2)$.
\end{proof}

%

\subsection{On the Directed Subgraph Homeomorphism Problem}

Theorem \ref{thm:vuln} and Algorithm \ref{alg} have consequences on another, long-standing problem:
the directed subgraph homeomorphism \cite{For1980}. The main problem is the following:
fixed a pattern graph $H$, determine whether $H$ is homeomorphic to a subgraph of any
given $G$ with respect to a given mapping from the nodes of $H$ to the nodes of $G$. This problem
is shown to be NP-hard, apart from a very simple family of pattern graphs for which a polynomial time
algorithm exists. Furthermore, in the conclusions the authors also discuss the problem when the mapping
is not given. Our characterisation falls in this latter case. 

In \cite{For1980} it is suggested that
the general problem is still NP-hard. However, in the special case when the pattern
graph has all nodes with indegree at most 1 and outdegree at most 2, or indegree at most 2 and 
outdegree at most 1, this is no
more necessary, since there exist pattern graphs of this kind for which a polynomial time algorithm
exists \cite{HU72}. However, as far as we know, no polynomial time algorithm for all this class of pattern
graphs has been devised so far.
Thus, our work also contributes to this research line: $\W$ is another pattern 
graph for which a polynomial time algorithm 
for the directed subgraph homeomorphism problem without node mapping
exists.

\begin{corollary}
Checking whether a multi-digraph $G$ contains a subgraph homeomorphic to $\W$
can be solved in time polynomial to the size of $G$.
\end{corollary}
\begin{proof}
Just run Algorithm \ref{alg} for every pair of distinct vertices in $G$ 
(ordinately playing the role of the source and target) 
and return true if and only if at least one of these calls returns true.
This algorithm costs $O(|V|^3\cdot |E|^2)$.
\end{proof}
\vspace{-2mm}

\section{Conclusions}
\label{conclusions}

\begin{figure}[t]
\begin{center}
\begin{tikzcd}[column sep=.6cm,row sep=.3cm]
&\\
s \ar[r] & u 
\ar[dd,bend left = 20]
\ar[r] & t\\
& & \\
& v \ar[uu,bend left = 20] & \\
\end{tikzcd}
\qquad\qquad
\begin{tikzcd}[column sep=.6cm,row sep=.5cm]
&\\
s \ar[r] & u \ar[r,bend left = 20]
& v \ar[l,bend left = 20] \ar[r] & t\\
&\\
&\\
\end{tikzcd}
\qquad\qquad
\begin{tikzcd}[column sep=.6cm,row sep=.6cm]
&\\
  s \ar[r] & u\ar[r]\ar[d] & v\ar[r] & t\\
               & x \ar[r,bend left = 20] & y \ar[l,bend left = 20]\ar[u]
\\
\end{tikzcd}
\vspace*{-.3cm}
\caption{Redundant Graphs that are not vulnerable}
\label{fig:incl}
\end{center}
\end{figure}
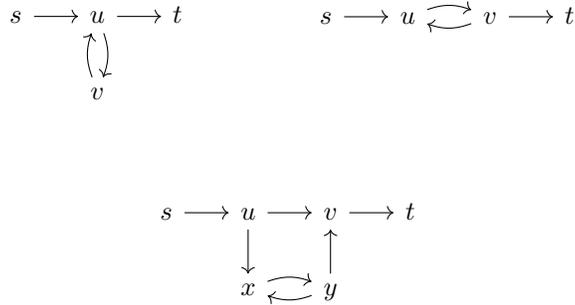

We have proved a graph-theoretic characterisation of vulnerable multi-digraphs that
generalizes the analogous one for undirected multigraphs \cite{Milch06}.
Our characterisation improves the results in \cite{ChenEtal15,CinesiFull} 
(given only for irreduntant directed multigraphs)
to general directed multigraphs; indeed, there exist redundant graphs
that are not vulnerable, as shown in Fig.\,\ref{fig:incl}.
Since graphs of this kind can easily appear in real traffic networks, we believe
that our characterisation was a necessary step for completing the knowledge
of vulnerability. 

Also the resulting algorithm was a necessary step: as we proved in this paper, 
checking irredundancy and calculating the maximal irredundant
subnet are NP-hard problems. 
Interestingly, a crucial part of our approach is the identification (and the consequent deletion)
of redundant edges. However, we want to remark that our algorithm neither identify {\em all} such edges 
in a graph nor can be used to decide if a {\em specific} edge is redundant.
Indeed, it removes all redundant edges only when it cannot find $\W$.
Thus, Algorithm 1 can solve (in polynomial time) \dered\ for non-vulnerable nets only. 

Of course, more efficient algorithms for checking
vulnerability should be devised
to be used in practice, but this was not our aim here.

\bibliographystyle{abbrv}
\bibliography{threeForms,chan}

\end{document}